\documentclass[lettersize,journal]{IEEEtran}
\usepackage{amsmath,amsfonts}
\usepackage{amssymb}
\usepackage{algorithmic}
\usepackage{algorithm}
\usepackage{array}
\usepackage{subcaption}
\usepackage{amsthm}
\usepackage{textcomp}
\usepackage{stfloats}
\usepackage{url}
\usepackage{verbatim}
\usepackage{graphicx}
\usepackage{flushend}
\usepackage{cite}
\usepackage{booktabs}
\allowdisplaybreaks
\newtheorem{remark}{\textbf{Remark}}
\newtheorem{observation}{\textbf{Observation}}

\begin{document}

\title{
Interference Propagation Analysis for Large-Scale\\ Multi-RIS-Empowered Wireless Communications:\\ An Epidemiological Perspective
}

\author{Kaining~Wang,~\IEEEmembership{Student Member,~IEEE,} Xueyao Zhang,~\IEEEmembership{Student Member,~IEEE,} Bo~Yang,~\IEEEmembership{Senior Member,~IEEE,} Xuelin Cao,~\IEEEmembership{Senior Member,~IEEE,} Qiang Cheng,~\IEEEmembership{Senior Member,~IEEE,} Zhiwen Yu,~\IEEEmembership{Senior Member,~IEEE,} Bin Guo,~\IEEEmembership{Senior Member,~IEEE,} George C. Alexandropoulos,~\IEEEmembership{Senior Member,~IEEE,} Kai-Kit Wong,~\IEEEmembership{Fellow,~IEEE}, Chan-Byoung Chae,~\IEEEmembership{Fellow,~IEEE}, and M\'erouane Debbah,~\IEEEmembership{Fellow,~IEEE} 
\thanks{
K. Wang, X. Zhang, B. Yang, and B. Guo are with the School of Computer Science, Northwestern Polytechnical University, Xi'an, Shaanxi, 710129, China (email: wangkaining, 2024263006@mail.nwpu.edu.cn, yang$\_$bo, guob@nwpu.edu.cn). 

X. Cao is with the School of Cyber Engineering, Xidian University, Xi'an, Shaanxi, 710071, China (email: caoxuelin@xidian.edu.cn). 

Q. Cheng is with the State Key Laboratory of Millimeter Waves,
Southeast University, Nanjing 210096, China (email: qiangcheng@seu.edu.cn).

Z. Yu is with the School of Computer Science, Northwestern Polytechnical University, Xi'an, Shaanxi, 710129, China, and Harbin Engineering University, Harbin, Heilongjiang, 150001, China (email: zhiwenyu@nwpu.edu.cn).


G. C. Alexandropoulos is with the Department of Informatics and Telecommunications, National and Kapodistrian University of Athens, 16122 Athens, Greece (email: alexandg@di.uoa.gr). 

K.-K. Wong is with the Department of Electronics
and Electrical Engineering, University College London, Torrington Place, WC1E7JE, United Kingdom and also affiliated with Yonsei Frontier Lab, Yonsei University, Seoul, 03722 Korea (e-mail: kai-kit.wong@ucl.ac.uk).

C.-B. Chae is with the School of Integrated Technology, Yonsei University, Seoul, 03722 South Korea (e-mail: cbchae@yonsei.ac.kr).


M. Debbah is with  KU 6G Research Center, Department of Computer and Information Engineering, Khalifa University, Abu Dhabi 127788, UAE and also with CentraleSupelec, University Paris-Saclay, 91192 Gif-sur-Yvette, France (email: merouane.debbah@ku.ac.ae)


}
}




\maketitle

\begin{abstract}
Reconfigurable intelligent surfaces (RISs) have gained significant attention in recent years due to their ability to control the reflection of radio-frequency signals and reshape the wireless propagation environment. Unlike traditional studies that primarily focus on the advantages of RISs, this paper examines the negative impacts of RISs by investigating interference propagation caused by user mobility in downlink wireless systems. We employ a stochastic geometric model to simulate the locations of base stations and RISs using the Mat\'{e}rn hard core point process, while user locations are modeled with the homogeneous Poisson point process. We derive novel closed-form expressions for the power distributions of the received signal at the users and the interfering signal. Additionally, we present a novel expression for coverage probability and introduce the concept of interference propagation intensity. To characterize the dynamics of interference caused by user mobility, we adopt an epidemiological approach using the susceptible-infected-susceptible model. 
Finally, crucial factors influencing the propagation of interference are analyzed. Numerical results validate our theoretical analysis and provide suggestions for managing interference propagation in large-scale multi-RIS wireless communication networks.
\end{abstract}

\begin{IEEEkeywords}
Reconfigurable intelligent surface (RIS), interference propagation, stochastic geometry, epidemic dynamics.
\end{IEEEkeywords}

\section{Introduction}

\IEEEPARstart{I}{n} recent years, reconfigurable intelligent surfaces (RISs) have attracted significant attention for their ability to manipulate radio-frequency (RF) signal propagation and reshape wireless communication environments \cite{RIS,hardware}. RIS is a planar metasurface with a controller and multiple low-cost, passive meta-atoms. These meta-atoms are composed of metal patches arranged in a dense configuration, allowing them to be independently controlled, or in groups, by an RIS controller. This independent control enables the intelligent adjustment of the phase, amplitude, or polarization characteristics of the reflected signals, thereby synergistically altering the wireless propagation environment. This technology has the potential to significantly enhance the efficiency and reliability of communication systems. It is characterized by high energy efficiency, low hardware costs, and ease of deployment, making it a game-changer for next-generation wireless networks \cite{RIS_challenges}. In this context, RISs are expected to be widely applied and deployed in future wireless communication networks to enhance coverage and reduce energy consumption.


\subsection{Related Work}
A significant amount of research has focused on analyzing the performance of RIS-assisted communication systems, predominantly examining link-level performance gains. Specifically, RIS and multi-access edge computing (MEC) technologies have been integrated into the space information network (SIN) to improve data rates and reduce latency~\cite{SINMEC}. 
In~\cite{RISISACSecure}, an RIS-assisted ISAC framework is proposed to achieve robust and secure transmission under imperfect sensing estimation.
The authors in~\cite{HAPWN} investigate a novel hybrid wireless network that includes active base stations (BSs) and passive RISs, demonstrating that deploying distributed RISs can effectively improve network throughput. The authors in~\cite{FGARIS} focus on the characteristics of millimeter-wave networks by incorporating directional beamforming and distinct path loss models for line-of-sight (LoS) and non-line-of-sight (NLoS) propagation. Furthermore, the authors in~\cite{RISAMCN} derive an expression for the ergodic rate of paired non-orthogonal multiple access (NOMA) users. 
The authors in ~\cite{RISMIMOAnalysis} obtain the cumulative distribution functions of large-scale parameters (LSPs) to explore the relationship between LSPs and heights.
However, these works overlook the significant impact that different association rules can have on system performance. To address this gap, the authors in~\cite{SGALIS} utilize stochastic geometry (SG) to analyze the performance of a large intelligent surface (LIS)-based millimeter-wave network, accounting for user association constraints in system deployment. In~\cite{PARLW}, closed-form expressions for the coverage probability under nearest-neighbor and fixed-association strategies in interference-limited scenarios are derived, demonstrating the effectiveness of RISs in enhancing system performance. Since the above studies primarily focus on flat obstacles and RISs, more realistic scenarios are considered in~\cite{CAIMW} and ~\cite{CARISMC}, where 2D buildings and 3D beamforming are involved to analyze network coverage.


 
Although the potential benefits of RIS have been well-documented, such as improved signal quality and coverage \cite{RIS-interference}, unintended consequences are often overlooked. {Compared to traditional communication systems, in fact, RIS-empowered communication systems may introduce additional interference due to the passive reflective nature of RIS} \cite{yb-tvt, 11212816}. Specifically, RIS typically receives control signals from BS and adjusts the reflection coefficients of each reflection element to achieve beamforming or scattering \cite{10600711}. However, for impinging signals from a non-associated BS \cite{10670007}, the RIS cannot achieve the desired signal reflection, leading to additional interference for the intended users. More critically, when we consider the user's movement \cite{9693982}, a mobile user may even turn its intended signal into interference to other adjacent users when they are in proximity, due to the imperfect beamforming of the RIS. \textit{{In this context, the interfering signals caused by an RIS can be propagated and ultimately degrade network performance, especially in large-scale multi-RIS wireless environments.}} 

\subsection{Motivation and Contributions}
This paper investigates the potential negative impact of deploying multiple RISs for downlink communications and studies the challenges associated with multi-RIS-induced interference propagation. Capitalizing on the powerful SG tool~\cite{SG2}, we represent the locations of the BSs and RISs using the Mat\'{e}rn Hardcore Point Process (MHCPP) \cite{MHCPP}, while user locations are modeled by the Homogeneous Poisson Point Process (HPPP) \cite{PPP}. 
We derive closed-form expressions for the power distribution of the desired and interfering signals within the RIS-assisted wireless system. To illustrate the interference propagation phenomenon, we apply an epidemic model to depict the number of infected nodes in the system \cite{epidemic}. The presented numerical results validate our theoretical model, highlighting the potential risks associated with RIS deployments and assessing key factors that determine RIS-induced interference. 

The main contributions of this paper are outlined as follows:
\begin{itemize}
    \item We present a framework for analyzing system-level interference in multi-RIS wireless communication systems. By employing the gamma approximation, we characterize the distributions of signal power and interference before and after user movement. Based on this analysis, the outage probability is derived.
    \item We study the phenomenon of interference propagation through an infectious disease model. By applying the SIS model, we classify mobile users as either susceptible or infected. Using the calculated outage probabilities, we derive closed-form expressions for the infection and recovery rates, allowing us to analyze the intensity of interference propagation within the system.
    \item By utilizing closed-form expressions, we analyze RIS-assisted interference propagation in large-scale wireless networks. We conduct Monte Carlo simulations to examine the effects of key variables on interference propagation, including BS and user densities, frequency, and the number of RIS elements. To minimize system interference and decrease the risk of interference propagation, it is crucial to deploy RIS strategically, especially in high-density scenarios.
\end{itemize}
\begin{figure}[t]
    \centering
    \includegraphics[width=0.98\linewidth]{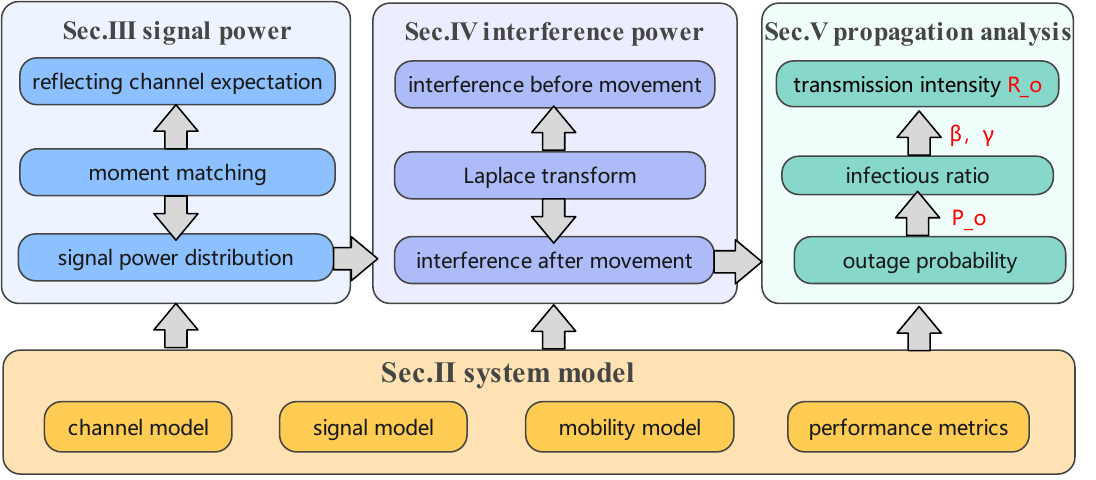}
    \caption{\small The structure of the paper.}
    \label{fig:structure}
\end{figure}

The remainder of this paper is structured as illustrated in Fig.~\ref{fig:structure}. In Sec.~\ref{system_model}, we introduce the SG-based system model. The signal power analysis and interference analysis are given in Sec.~\ref{power_analysis} and Sec.~\ref{inter_analysis}, respectively. Sec.~\ref{inter_pro_analysis} details the analysis of interference propagation, while Sec.~\ref{results} presents and discusses the numerical results. Finally, the paper is concluded in Sec.~\ref{conculsion}.

\begin{table*}[]
\renewcommand\arraystretch{1.25} 
\centering
\captionsetup{font={small}}
\caption{The notations of the paper.}
\label{tab:t1} 
\resizebox{0.9\textwidth}{!}{  
\begin{tabular}{|c|l|}    
\hline
\small
\textbf{Notation}&\textbf{Definition}\\
\hline\hline
\textit{P}&Transmit power of the BSs\\
\hline
$\Theta_B$,$\Theta_R$,$\Theta_U$&Set of the BSs/RISs/UEs\\
\hline
$\lambda_B,\lambda_R,\lambda_U$&Density of the BSs/RISs/UEs\\
\hline
$PL_{ij},PL_{ik},PL_{jk},PL_{ijk}$ &Path loss between the BS and RIS, the BS and UE, the RIS and UE, the BS and RIS and UE\\
\hline
$\alpha$&Path loss exponent of the channel\\
\hline
$C$&Path loss of the channel at the reference distance of 1 meter\\
\hline
$N$&Number of elements for each RIS\\
\hline
$h_{ik}$&Small-scale fading of the direct link\\
\hline
$h_{ij},h_{jk}$& Nakagami-$m$ channel fading coefficient of the reflected link\\
\hline
$\phi_k$&Phase shift matrix of the RIS\\
\hline
$d_{ij},d_{ik},d_{jk},d_{ijk}$&Distance between the BS and RIS, the BS and UE, the RIS and UE, the BS and RIS and UE\\
\hline
$T$ & Threshold for the signal-to-interference-plus-noise ratio (SINR)\\
\hline
$\beta,\mu$&Infection rate and recovery rate of the SIS epidemic model\\
\hline
\end{tabular}
}
\end{table*}

\section{System Model}
\label{system_model}
\begin{figure}
    \centering
    \includegraphics[width=\linewidth]{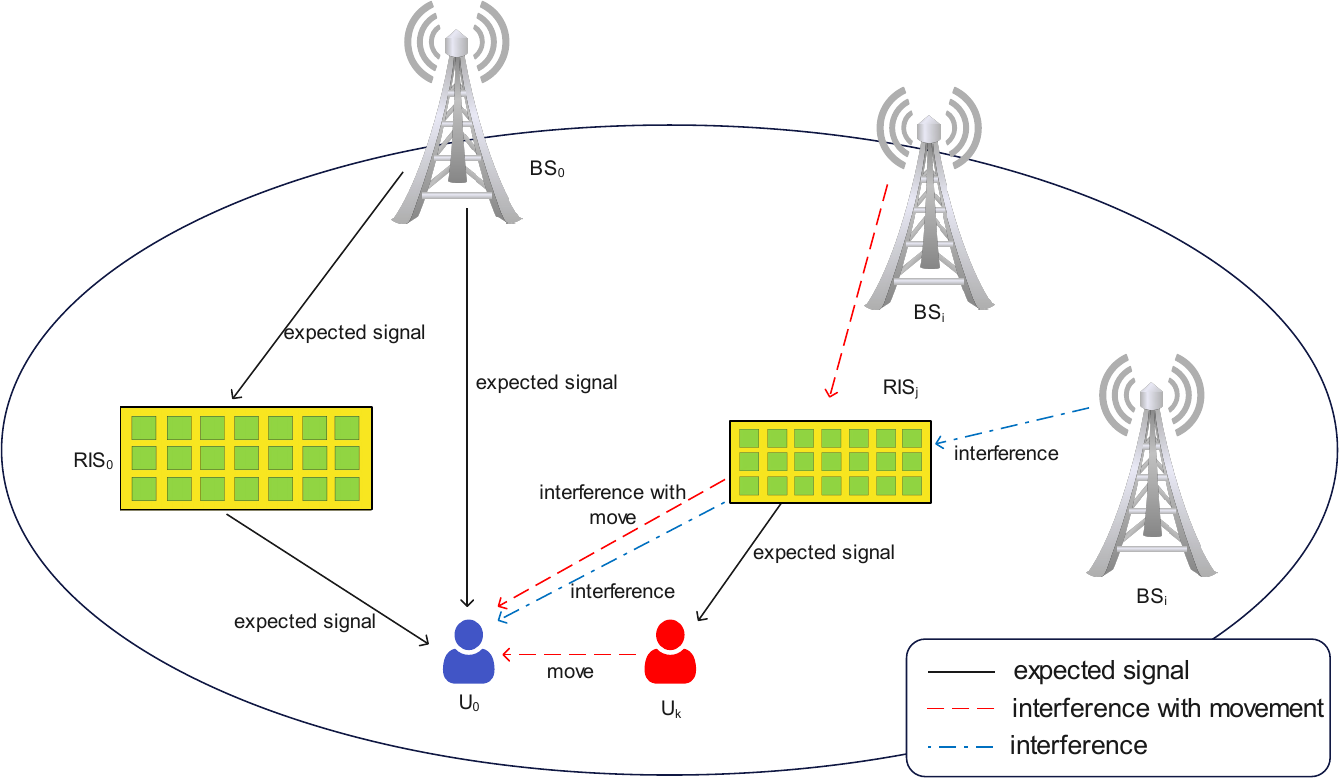}
    \caption{\small The considered multi-RIS-assisted multi-cell downlink wireless system. A typical user $U_0$ serves as the observation target in our interference propagation analysis. Each BS owns a dedicated RIS~\cite{10670007}, enabling both direct and RIS-reflected transmission paths. Solid lines represent the desired signal paths received by $U_0$,  including both direct BS-UE links and RIS-reflected BS-RIS-UE links. Blue dashed lines denote interference originating from interfering BSs and their associated RISs. Red dashed lines represent additional interference caused by the movement of other users, whose reflected signals via nearby RISs unintentionally contaminate $U_0$'s reception, thereby propagating interference dynamically.}
    \label{fig:scene}
\end{figure}
We consider a multi-RIS-assisted multi-cell cellular network comprising $B$ single-antenna BSs, $R$ RISs, and $U$ single-antenna user equipment (UE) nodes. We denote the sets of BSs, RISs, and UEs as $\Theta_{B}=\{1,2,...,B\}$, $\Theta_{R}=\{1,2,...,R\}$, and $\Theta_{U}=\{1,2,...,U\}$, respectively. Let $i \in \Theta_{B}$, $j \in \Theta_{R}$, and $k \in \Theta_{U}$ denote the indices of a BS, a RIS, and a UE, respectively. In this paper, we analyze the downlink communication performance from the BSs to the UEs. All RISs comprise reflective elements arranged as a rectangular array. As demonstrated in~\cite{11078147}, a two-bit distinct phase states per RIS element can yield reflection performance very close to the theoretical upper bound. Hence, we assume, in this paper, this resolution for all metematerials of all RISs. We establish a right-angle coordinate system with the ground as the \textit{xOy} plane, where the unit vector components of the \textit{x}-axis and \textit{y}-axis are $\mathbf{e_x}$ and $\mathbf{e_y}$, respectively. For the locations of the BSs and RISs, we use the hard-core cluster process. The BS position is the parent process, and the RIS position is the child process. The BS acts as the cluster center, which uses the MHCPP to ensure that the distance between the BSs is greater than $r_B$. This process can be regarded as excluding all nodes that violate the hard-core condition based on the HPPP and retaining nodes whose distance is greater than or equal to $r_B$. We consider deploying distributed RISs to assist BS-UE communication in the network. It is assumed that all RISs have the same height $H_R$ in meters. It is noted that if the RISs are too far away from the BS, they will have difficulty receiving control signals from the BS. Therefore, the RISs obey the HPPP distribution and they are distributed in a circle with the BS as the cluster center and radius $r_R$. The locations of the UEs constitute an independent HPPP with density $\lambda$. We consider a typical user, $U_0$, located at the origin, which serves as the observation target in the interference analysis. According to Slivnyak's theorem \cite{SG3}, the performance of $U_0$ represents the average performance of the UE. All our system model parameters are appended in Table~\ref{tab:t1}.

In our system model, we assume that each UE is associated with only one BS. This connection can be established through direct or indirect links. A direct link connects a BS and a UE, while an indirect link connects them via a RIS reflection. We assume that each UE selects the nearest BS to establish a connection, and similarly, each RIS is associated with the nearest BS. We also assume that reflected links can form only through an RIS, without considering links through multiple RISs.

\subsection{Channel Model}
We denote the transmitted power as $P$ and the regularized small-scale fading channel as $g$. The direct link is assumed to follow a Rayleigh fading channel, suitable for scattering-rich urban environments, and no stable LoS path exists between the BS and its assigned UE. We use the product-distance model \cite{productPL} as our path-loss (PL) model. To this end, the PL of the direct link is expressed as follows:
\begin{equation}
    PL_{ik}=Cd_{ik}^{-\alpha },
\end{equation}
where $C$ denotes the free-space PL coefficient, incorporating the carrier wavelength $\lambda$, the transmit antenna gain $G_t$, the receive antenna gain $G_r$, and the free-space path index. In particular, $C=\left(\frac{\lambda \sqrt{G_tG_r}}{4\pi}\right)^2$, where $d$ is the distance between the UE and BS, and $\alpha$ is the PL exponent. It is assumed that the horizontal distances between channels are much larger than the vertical distances, and that the heights of the BSs, RISs, and UEs are negligible.

Similarly, the PL of the BS-RIS-UE (reflected link) can be expressed as follows:
\begin{equation}
    PL_{ijk}=C(d_{ij}d_{jk})^{-\alpha},
\end{equation}
where $d_{i,j}$ is the distance between the BS and RIS, $d_{jk}$ is the distance between the RIS and UE, and $C$ is the per unit PL of reflected link. 


\subsection{Signal Model}
Assuming that each BS is associated and controls a distinct RIS, the received signal at a UE $k$ can be expressed as follows:
\begin{equation}
    y_k=\left(\sqrt{PL_{ik}}h_{ik}+\sqrt{PL_{ijk}}h_{ij}\Phi_kh_{jk}\right)\sqrt{P}s_k+I_k+n,
\label{con:yk}
\end{equation}
where $\Phi_k = diag(\boldsymbol{\theta_k})$ denotes the diagonal reflection matrix of the RIS assisting the downlink communication towards the $k$-th UE; $s_k$ represents this UE's unit power transmit symbol from BS; $h_{ij}$ and $h_{jk}$ represent the normalized small-scale fading of the reflected channels; $n\sim CN(0,\epsilon^2I_N)$ is the additive white Gaussian noise (AWGN) at the BS with zero mean and covariance matrix $\epsilon^2I_N$, where $\epsilon^2$ is the equivalent noise power and $I_N$ is the identity matrix of size $N$. To enable coherent signal combining at the UE side, the RIS adjusts the phase shifts of its reflecting elements under the control of the associated BS~\cite{10600711}. Specifically, the phase of the reflected path is aligned such that it adds constructively with the direct path at the receiver, thereby maximizing the total received signal power. The optimal RIS phase shift for UE \textit{k} is designed to compensate for the phase differences along the BS-RIS and RIS-UE links, i.e.:
\begin{equation}
    \boldsymbol{\theta_k}=-arg(h_{ij,n})-arg(h_{jk,n})+arg(h_{ik,n}), n=1,...,N.
\end{equation}
In \eqref{con:yk}, the term $I_k$ represents the interference from all other BSs and can be expressed as follows:
\begin{equation}
    I_k=\sum_{i\in B \backslash \{i(k)\}} (\sqrt{PL_{ik}}h_{ik}+\sqrt{PL_{ijk}}h_{ij}\Phi_kh_{jk})\sqrt{P}s_k.
\end{equation}
The interference received by UE $k$ comes from all BSs except for its serving BS $i(k)$. $i(k)$ denotes the index of the serving BS for user $k$, which is selected based on the nearest BS association strategy:
\begin{equation}
i(k) = \arg\min_{i \in \Theta_B} d_{ik}.
\end{equation}

Thus, the SINR at the \textit{k}-th UE is given by:

\begin{equation}
	\gamma_k\!=\!\frac{P\left|\sqrt{PL_{ik}}h_{ik}\!+\!\sqrt{PL_{ijk}}h_{ij}\Phi_kh_{jk}\right|^2}{\sum_{i\in B \backslash \{i(k)\}}P\left|\sqrt{PL_{ik}}h_{ik}\!+\!\sqrt{PL_{ijk}}h_{ij}\Phi_kh_{jk}\right|^2\!+\!\sigma^2},
\end{equation}
where the numerator includes the power of the intended signal at the \textit{k}-th UE, and the denominator includes the interference from all other BSs in the system.

\subsection{Mobility Model}
In RIS-empowered downlink communication scenarios~\cite{RIS_challenges}, the BS is fixed, and the RIS is statically deployed near it. However, the UE's location is random and unpredictable. We adopt the random walk mobility model \cite{mobility} to describe the movement of the UEs in the system. According to this model, a UE moves from its current location ($L$) to a new location ($L'$) by randomly selecting a direction (from [0, 2$\pi$]) and a distance (from [0, 10] meters). Given that the RIS has beamforming capabilities and can be optimized to track mobile UEs, a UE experiencing interference will move, thereby propagating the interference. This implies that if this UE moves close to another UE, the latter's interference level will be impacted.

\subsection{Epidemic Model}
The epidemic model \cite{epidemicModel} examines the emergence of infectious diseases, their spread within populations, and strategies for preventing, controlling, and mitigating their spread. These models have also been explored to describe the spread of computer viruses across networks, such as the Internet \cite{EpidemicInternet} and emails \cite{EpidemicEmail}. In large-scale networks, the dissemination of interference due to UE mobility exhibits behavior analogous to the spread of infectious diseases. Therefore, in this paper, we use an epidemic model to assess the impact of interference propagation due to UE mobility in an RIS-assisted communication system. Common infectious models include SI, SIS, SIR, and SEIR models \cite{SI, SIS2, SIR, SEIR}, which respectively model different types of diseases. We specifically utilize the SIS model \cite{sis} to describe the dynamics of interference signal propagation over a pre-defined range. In the SIS model, UEs are represented as distinct nodes, each existing in one of two states: susceptible (S), meaning the UE is vulnerable to interference; or infected (I), indicating that the UE has been adversely affected by interference. In this context, ``infected'' refers to the SINR of UE $k$ being below a certain threshold, i.e., $\gamma_k<T$. We denote the number of susceptible individuals by $S(t)$, and the number of infected individuals by $X(t)$. Assuming the total number of nodes remains constant, we have $S(t) + X(t) = N_{U}$, where $N_U$ denotes the total number of UEs. In the SIS model, susceptible UEs become infected through contact with infected UEs and can be reinfected after recovery, as there is no immunity. The dynamic process of infection propagation within a given area can be modeled as follows:
\begin{equation}
\begin{cases}
 \frac{dS(t)}{dt}=-\beta X(t)S(t)+\mu X(t) ,
 \\\frac{dX(t)}{dt}=\beta X(t)S(t)-\mu X(t),
\end{cases}
\end{equation}
where $\beta$ represents the probability of infection through close proximity in a unit of time. To this end, $\beta s(t)$ denotes the average number of newly infected UEs generated by an infected UE through contact transmission in a unit of time. Finally, $\beta i(t)s(t)$ denotes the number of newly infected UE in a unit of time, and $\mu$ is denoted as the recovery rate.

\subsection{Performance Metrics}
In this paper, we analyze the coverage probability and interference propagation intensity for the considered RIS-empowered downlink communication system.

1) \textbf{Outage Probability:} The outage probability is defined as the probability that the received SINR is below a threshold:
\begin{equation}\label{PoD}
    P_o(\gamma_k)={\rm Pr}\left[\gamma_k<T\right].
\end{equation}
Given the outage probability before UE movement as $P_o$, and after movement as $P_o'$, the infection rate $\beta$ can be modeled as:
\begin{equation}
    \beta =\left(1-P_o\right)P_o'.
\end{equation}
Consequently, te recovery rate $\mu$ is given by:
\begin{equation}
    \mu =P_o\left(1-P_o'\right).
\end{equation}

2) \textbf{Interference Propagation Intensity:} The interference propagation intensity, $R_0$, is defined as the ratio of the infection rate over the recovery rate. In our SIS model, if $R_0 > 1$ holds, each infected UE will, on average, infect more than one UE, so the disturbance will continue to spread in the system. Otherwise, if $R_0 \leq 1$, then each UE's infection will not be sufficient to maintain the spread of the disturbance, so the system disturbance will gradually subside.

\section{Signal Power Analysis}
\label{power_analysis}
In this section, we analyze the average signal power distribution, which has a more complex form than in the traditional case without any RIS. We assume that each UE receives both the direct-link signal and the RIS-assisted-link signal. Then, the RIS-assisted received signal power is given by:
\begin{equation}\label{eq:S_0}
S_0=\left(\sqrt{PL_{ik}}|g_0|e^{\phi_{ik}}+\sqrt{PL_{ijk}}S_r\right)^2,
\end{equation}
where $S_r$ denotes the gain of the BS-RIS and RIS-UE channels, which can be expressed as follows:
\begin{equation}
    \begin{split}
        S_r&=\sum_{n=1}^{N}|h_{ij}|\phi_j|h_{jk}|=\sum_{n=1}^{N}|h_{ij}||h_{jk}|e^{\phi_{ij}+\phi_{jk}+\phi_k}.
    \end{split}
\end{equation}

To maximize the received signal power, the RIS phase shifts are designed such that the reflected path adds coherently with the direct path. This requires aligning the phases of the two-channel components. The phase shifts can be designed as\footnote{Note that, if the goal was to maximize the desired signal power and simultaneously minimize the interference power, a different RIS design would be needed. In that case, estimating interference-related parameters would impose greater system overhead.} To facilitate phase alignment at the receiver, the RIS phase shift $\theta_{k,n}$ applied at the $n$-th element is designed based on the phases of the involved channels. Letting $\phi_{ij}=arg(h_{ij,n})$, $\phi_{jk}=arg(h_{jk,n})$, and $\phi_{ik}=arg(h_{ik,n})$, the optimal phase shift is: $\theta_{k,n}=-\phi_{ij}-\phi_{jk}+\phi_{ik}$. For compactness, we denote the total effective RIS phase shift as $\phi_{k}=\theta_{k,n}$, and use this in the refined expression of \eqref{eq:S_0}.
\begin{equation}
S_0=\left(\sqrt{PL_{i0}}|h_{i0}|+\sqrt{PL_{ij0}}\sum_{n=1}^{N}|h_{ij}||h_{j0}|\right)^2.
\end{equation}

Since the exact distribution of $S_r$ is difficult to compute in closed form, we deploy the moment matching method \cite{mmm} to derive the distribution of the reflected signal. Since both $h_{ij}$ and $h_{jk}$ follow the Nakagami-$m$ distribution \cite{nakagami}, the expectation of $h_{ij}$ and $h_{jk}$ can be given as follows:
\begin{equation}
    \mathbb{E}\left\{|h_{ij}|\right\}=\frac{\Gamma(m_1+1/2)}{\Gamma(m_1)}\frac{1}{\sqrt{m_1}},
\end{equation}
\begin{equation}
    \mathbb{E}\left\{|h_{jk}|\right\}=\frac{\Gamma(m_2+1/2)}{\Gamma(m_2)}\frac{1}{\sqrt{m_2}},
\end{equation}
where $m_1$ and $m_2$ denote the $m$-parameters of the BS-RIS and RIS-UE channels. Hence, the first and second order moments of $S_r$ can be computed as follows:
\begin{equation}
    \begin{split}
        \mathbb{E}\left\{S_r\right\} &= \sum_{n=1}^{N}\mathbb{E}\left\{|h_{ij,n}|\right\}\mathbb{E}\left\{|h_{jk,n}|\right\}\\
    &=N\mathbb{E}\left\{|h_{ij,n}|\right\}\mathbb{E}\left\{|h_{jk,n}|\right\}\\
    &=N\frac{\Gamma(m_1+1/2)}{\Gamma(m_1)}\frac{\Gamma(m_2+1/2)}{\Gamma(m_2)}\frac{1}{\sqrt{m_1m_2}},
    \end{split}
\end{equation}
\begin{equation}
    \begin{split}
        \mathbb{E}\left\{S_r^2\right\}&=\mathbb{E}\left\{\sum_{n=1}^{N}|h_{0,n}^2|r_{0,n}|^2\right\}\\
    &\hspace{0.4cm}+2\mathbb{E}\left\{\sum_{n=1}^{N-1}\sum_{i=n+1}^N|h_{0,n}||r_{0,n}||h_{0,i}|r_{0,i}\right\}\\
    &=N+N(N-1)\\
    &\hspace{0.4cm}\times\left(\frac{\Gamma(m_1+1/2)}{\Gamma(m_1)}\frac{\Gamma(m_2+1/2)}{\Gamma(m_2)}\right)^2\frac{1}{m_1m_2}.
    \end{split}
\end{equation}
Clearly, $S_r$ can be parameterized by the shape and scale of the Gamma random variable as:
\begin{align}
    \kappa _r=\frac{\mathbb{E}^2\{S_r\}}{{\rm Var}\{S_r\}} ,\eta_r=\frac{{\rm Var}\{S_r\}}{\mathbb{E}\{S_r\}}, 
\end{align}
where ${\rm Var}\{S_r\}=\mathbb{E}\{S_r^2\}-\mathbb{E}^2\{S_r\}$. 

Since the distribution of the Gamma random variable $S_0$ consists of the random variables $g_0$ and $S_r$, we need not only the first and second order moments of $S_r$, but also the first and second order moments of the direct link $g_0$. Assuming that the direct link is a Rayleigh fading channel, its first and second order moments are well known to be given by:
\begin{align}
    &E\left\{|g_0|\right\}=\sqrt{\frac{\pi }{4}},\,\,{\rm Var}\left\{|g_0|\right\}=1-\frac{\pi}{4}.
\end{align}
Following the solution procedure in~\cite{PARLW}, we can obtain a Gamma-based approximation for $S_0$ using the first and second order moments in (\ref{con:s0}) and (\ref{con:s02}) (top of next page).
To this end, its shape, $k_s$, and scale, $\eta_s$, parameters can be straightforwardly derived as follows:
\begin{figure*}[t]
\begin{equation}
\begin{split}
    \mathbb{E}\{S_0\}&=PL_{i0}\mathbb{E}\{|g_0|^2\} + 2\sqrt{PL_{i0}PL_{ij0}}\mathbb{E}\{|g_0|\}\mathbb{E} \{S_r\}+PL_{ij0}\mathbb{E}\{S_r^2\}\\
    &=\frac{\pi PL_{ij}}{4}+2N\sqrt{\frac{\pi PL_{ij}PL_{ij0}}{4} }\frac{\Gamma(m_1+1/2)}{\Gamma(m_1)}\frac{\Gamma(m_2+1/2)}{\Gamma(m_2)}\frac{1}{\sqrt{m_1m_2}}+N\\
    &\hspace{0.4cm}+N(N-1)\left(\frac{\Gamma(m_1+1/2)}{\Gamma(m_1)}\frac{\Gamma(m_2+1/2)}{\Gamma(m_2)}\right)^2\frac{1}{m_1m_2}
\end{split}
\label{con:s0}
\end{equation}
\end{figure*}

\begin{figure*}[ht]
\begin{equation}
\begin{split}
    \mathbb{E}\{S_0^2\}&=\mathbb{E}\left\{PL_{ik}^2g_{0}^4+PL_{ijk}^2Sr^4+4PL_{ik}PL_{ijk}g_{0}^2Sr^2+2PL_{ik}PL_{ijk}g_{0}^2Sr^2\right.\\
&\left.\hspace{0.4cm}+4\sqrt{PL_{ik}^3}PL_{ijk}g_{0}^3Sr+4\sqrt{PL_{ik}PL_{ijk}^3}g_{0}Sr^3\right\}\\
&=PL_{ik}^2\left(\Gamma(3)+4\sqrt{PL_{ijk}/PL_{ik}}\Gamma(5/2)\eta _r\frac{\Gamma(1+\kappa_r)}{\Gamma(\kappa _r)}    +6\frac{PL_{ijk}}{PL_{ik}}\Gamma(2)\eta _r^2 \frac{\Gamma(2+\kappa_r)}{\Gamma(\kappa _r)}\right.\\ &\left.\hspace{0.4cm}+4\sqrt{PL_{ijk}/PL_{ik}^3}\Gamma(3/2)\eta _r^3\frac{\Gamma(3+\kappa _r)}{\Gamma(\kappa _r)}
    +\frac{PL_{ijk}}{PL_{ik}}^2 \eta _r^4 \frac{\Gamma(4+\kappa _r)}{\Gamma(\kappa _r)}\right)
\end{split}
\label{con:s02}
\end{equation}
\hrulefill
\end{figure*}
\begin{align}
    k_s=\frac{\mathbb{E}\{S_0\}^2}{Var\{S_0\}},\,\,
    \eta_s=\frac{{\rm Var}\{S_0\}}{\mathbb{E}\{S_0\}}.
\end{align}
Hence, the probability density function and cumulative distribution function of $S_0$ can be expressed directly as:
\begin{align}
    f_{S_0}(x)&=\frac{x^{k_s-1}\exp\left(-x/\eta _s\right)}{\Gamma \left(k_s\right)\eta _s^{k_s}},\\
    F_{S_0}(x)&=\frac{1}{\Gamma(k)}\gamma\left(k_s,x/\eta\right),
\end{align}
where $\gamma (\cdot,\cdot)$ is the lower incomplete gamma function.

\section{Interference Power Analysis}
\label{inter_analysis}
In this section, we analyze the distribution of the interference power. Our model assumes that interference signals at each UE originate from all BSs except the designated serving BS. 
Thus, the interference signals are subject to misaligned phase shifts that may alter their negative impact on the intended UE~\cite{Interference1}. 

It is well known that RISs can operate in various modes~\cite{hardware}, e.g., beam steering, beamforming/focusing, and beam splitting,   
each having its unique characteristics. Specifically, beam steering directs the signal towards the target user, while beamforming converges multiple signals into a single focused beam. Beam splitting divides a single incoming signal into multiple outgoing signals in different directions. Regardless of the mode used, the phase coupling between the interference and the RIS can boost interference power. When considering UE mobility \cite{Usermobility}, if beam deflection occurs, the interference will follow the intended UE's movement, but this may not significantly alter the overall interference level. In beamforming, because the RIS passively reflects signals without distinguishing between useful and interfering signals, the latter are directed toward the target UE, resulting in interference boosting. It is important to note that our analysis assumes perfect beamforming without errors. Any phase shift errors \cite{phaseError} can exacerbate the level of interference. Additionally, as other UEs approach, their desired signals may become interfering signals for the target UE.

To address the latter situation, we need to examine the distribution of interference both before and after the UE performs a movement, comparing with and without UE mobility. Here, we focus on how interfering signals are received in a communication environment assisted by multiple RISs.

\subsection{Interference before UE Movement}
We first evaluate the interference distribution when no UE movement occurs. In this case, it holds:
\begin{equation}
        I_k=\sum_{i\in B \backslash \{i(k)\}}\left(\sqrt{PL_i}h_{ik}+\sum_{j\in R}\sqrt{PL_{ijk}}h_{ij}\Phi_jh_{jk}\right)^2.
\end{equation}
To compute the distribution of the aggregated interference $I_k$ at the target UE, we employ the Laplace transform \cite{laplace}, which is a standard and tractable tool in SG-based wireless network analysis. Specifically, the Laplace transform of $I_k$ enables us to derive closed-form expressions for key performance metrics, including outage and coverage probabilities. This is particularly useful when $I_k$ is a sum of many random variables stemming from spatially distributed interferers, making its probability density function intractable to obtain directly. To this end, the Laplace transform of the interference term $I_k$ experienced by the target UE $k$ can be expressed as:
\begin{equation}\label{eq:l1}
\begin{split}
    &\mathcal{L}_{I_k}(s)=\exp\left( -2\pi\lambda_B\int^\infty\left(1-\frac1{1+sC\nu^{-\alpha}} \right)\nu d\nu\right)\\
    &\times\exp\left(-2\pi\lambda_B\int^\infty\left(1-\mathbb{E}_{R}\left\{\prod_{j\in R}\frac1{sNC^2\nu^{-\alpha}d_{jk}^{-\alpha}+1}\right\}\right)\nu d\nu\right)\\
    &=\exp\left(\frac{ -2\pi^2 \lambda_B{\rm csc}\left(\frac{2\pi}{\alpha }\right)(sC)^{2/\alpha} }{\alpha}\right)\\
    &\times \exp\left(-2\pi\lambda_B\left(\frac{2\pi^2\lambda_Rcsc(2\pi/\alpha )sN^2C^2 }{\alpha^2 }\left(\ln_{}{d_{max}}\!-\!\ln_{}{d_{min}}\right)\right.\right.\\
    &\left.\left.+\frac{\alpha }{\alpha -1}\left(d_{max}^{(1-1/\alpha)}\!-\!d_{min}^{(1-1/\alpha)}\right)\right)\right),
\end{split}
\end{equation}
where $d_{min}$ and $d_{max}$ denote the minimum and maximum values between the UE and RIS, respectively. $\mathbb{E}_{R}$ is the expectation taken over the spatial distribution of RISs, modeled as a PPP. And $\csc(x)=1/\sin(x)$ is the cosecant function.

\begin{proof}
    See Appendix A.
\end{proof}

\subsection{Interference after UE Movement}
Unlike most existing work, when UE movement is considered, the interference caused by RIS must be modeled. This is important because the RIS may direct interfering signals from other BSs to the target UE. Additionally, we should consider interference from directly impacting BSs, as well as interference reflected from the BSs via the RIS (the one assigned to serve the target UE). Consequently, as a UE approaches the target UE, the interference signal can be expressed as follows:
\begin{equation}
    \begin{split}
        I_k'=&\sum_{i\in B \backslash \{i(k)\}}\left(\sqrt{PL_{ik}}h_{ik}\!+\!\sum_{j\in R}\sqrt{PL_{ijk}}h_{ij}\Phi_jh_{jk}\right)^2\\
    &\!+\!\sum_{k\in k_I}\left(\sqrt{PL_{ijk}}h_{ij}\Phi_jh_{jk}\right)^2,
    \end{split}
\end{equation}
where $k_I$ denotes the number of UEs causing interference and
\begin{equation}
    h_{ij}\Phi_jh_{jk} = \sum_{n=1}^{N}h_{ij}^ne^{i\phi_{jt}}h_{jk}^n.
\end{equation} 

Obviously, the expression of the disturbance is in the form of a multivariate sum, and thus its statistical properties can be obtained by using the Laplace transform to convert it into a natural exponential function. 

Therefore, the statistical characterization of the interference power $I'$ for the target UE is given by:
\begin{align}\label{eq:l2}
    &\mathcal{L}_{I_k'}(s)=\exp\left(\frac{ -2\pi^2 \lambda_B{\rm csc}\left(\frac{2\pi}{\alpha }\right)(sc)^{2/\alpha} }{\alpha }\right)\nonumber\\
    &\times \exp\left(-2\pi\lambda_B\left(\frac{2\pi^2\lambda_Rcsc(2\pi/\alpha )sN^2c^2 }{\alpha^2 }\right.\right.\nonumber\\
    &\left.\left.\times\left(\ln_{}{d_{max}}\!-\!\ln_{}{d_{min}}\right)
    +\frac{\alpha }{\alpha\!-\!1}\left(d_{max}^{(1\!-\!1/\alpha)}\!-\!d_{min}^{(1\!-\!1/\alpha)}\right)\right)\right)\nonumber\\
    &\times 
    \exp\left(\frac{ -2\pi^2 \lambda_{U_{near}}{\rm csc}(\frac{2\pi}{\alpha })(sc)^{2/\alpha} }{\alpha}\right),
\end{align}
where $\lambda_{U_{near}}=\lambda_{B}\lambda_{U}\pi r_I^2$ with $r_I$ being the interference area. 

\begin{proof}
    See Appendix B.
\end{proof}

\begin{remark}
To address interference caused by UEs' movement, we focus solely on the effects and frequency of movement for the target user within a specific time frame. Additionally, it is important to note that this interference originates from the target UE or the BS associated with it.
\end{remark}

\section{Interference Propagation Analysis}
\label{inter_pro_analysis}
In this section, we analyze the outage probability, infection rate, and interference propagation intensity of the target UEs.
\subsection{Outage Probability}
\begin{enumerate}
    \item Outage probability before UE movement: If $U_0$ is only associated with the nearest BS at a fixed location, the outage probability is derived as follows:
    \begin{equation}\label{Op1}
    P_o=1-\sum_{x  = 0}^{k_s-1}\frac{(-1)^x}{x!}\frac{\partial ^x}{\partial ^{s^x}}\exp\left(\frac{-sT{\sigma^{2} }}{{P\eta}}\right)\mathcal{L}_{I_k}\left(\frac{sT}{\eta }\right).
    \end{equation}
   \begin{proof}
    See Appendix C.
\end{proof}
    \item Outage probability after UE movement: 
    \begin{equation}\label{Op2}
        P_o'=1-\sum_{x  = 0}^{k_s-1}\frac{(-1)^x}{x!}\frac{\partial ^x}{\partial ^{s^x}}\exp\left(\frac{-sT{\sigma^{2} }}{{P\eta}}\right)\mathcal{L}_{I_k'}\left(\frac{sT}{\eta }\right).
    \end{equation}
\end{enumerate}
   \begin{proof}
   The proof is similar to that for $P_o$ above.
\end{proof} 
\begin{remark}
By substituting expressions (\ref{eq:l1}) and (\ref{eq:l2}) into (\ref{Op1}) and (\ref{Op2}), we can derive the outage probability expressions. It is evident that, as the interference power increases, the density of BSs and RISs leads to a higher outage probability.
\end{remark}

\subsection{Infection and Recovery Rates}
The infection rate is the likelihood of transitioning from an uninfected to an infected state before and after a UE moves. Specifically, it denotes the probability that the SINR falls from above the threshold to below it due to UE mobility. To this end:
\begin{align}
    \beta&=(1-P_o)P_o'\nonumber\\
    &=\sum_{x  = 0}^{k_s-1}\frac{(-1)^x}{x!}\frac{\partial ^x}{\partial ^{s^x}}\exp\left(\frac{-sT{\sigma^{2} }}{{P\eta}}\right)\mathcal{L}_{I_k}\left(\frac{sT}{\eta _S}\right)\nonumber\\
    &\hspace{0.3cm}-\sum_{y=0}^{k_s-1}\sum_{x = 0}^{k_s-1}\frac{(-1)^x}{x!}\frac{(-1)^y}{y!}\frac{\partial ^x}{\partial ^{s^x}}\exp\left(\frac{-sT{\sigma^{2} }}{{P\eta}}\right)\mathcal{L}_{I_k}\left(\frac{sT}{\eta _S}\right)\nonumber\\
    &\hspace{0.3cm}\times\frac{\partial ^y}{\partial ^{s^y}}\exp\left(\frac{-sT{\sigma^{2} }}{{P\eta}}\right)\mathcal{L}_{I_k'}\left(\frac{sT}{\eta _S}\right).
\end{align}

\begin{remark}
    Several factors can cause a UE to change from an uninfected state to an infected state, including the target UE moving into an area with severe interference, an increase in the power of any interfering BS, or the movement of an interfering UE that propagates the interference.
\end{remark}

The recovery rate describes the change from an infected state to a susceptible state before and after the UE moves. Specifically, it refers to the probability that the SINR, initially below the threshold, exceeds it after the movement. To this end:
\begin{align}
    \gamma&=P_o(1-P_o')\nonumber\\
    &=\sum_{x = 0}^{k_s-1}\frac{(-1)^x}{x!}\frac{\partial ^x}{\partial ^{s^x}}\exp\left(\frac{-sT{\sigma^{2} }}{{P\eta}}\right)\mathcal{L}_{I_k'}\left(\frac{sT}{\eta _S} \right)\nonumber\\
    &\hspace{0.3cm}-\sum_{y=0}^{k_s-1}\sum_{x  = 0}^{k_s-1}\frac{(-1)^x}{x!}\frac{(-1)^y}{y!}\frac{\partial ^x}{\partial ^{s^x}}\exp\left(\frac{-sT{\sigma^{2} }}{{P\eta}}\right)\mathcal{L}_{I_k'}\left(\frac{sT}{\eta _S} \right)\nonumber\\
    &\hspace{0.3cm}\times \frac{\partial ^y}{\partial ^{s^y}}\exp\left(\frac{-sT{\sigma^{2} }}{{P\eta}}\right)\mathcal{L}_{I_k}\left(\frac{sT}{\eta _S}\right).
\end{align}


\subsection{Interference Propagation Intensity}
Having derived the infection and recovery rates, we need to consider these two variables together to obtain a measure of the rate of interference propagation, as follows:
\begin{equation}
\begin{split}
    R_0=\frac{1/\sum_{x  = 0}^{k_s-1}\frac{(-1)^x}{x!}\frac{\partial ^x}{\partial ^{s^x}}\exp\left(\frac{-sT{\sigma^{2} }}{{P\eta}}\right)\mathcal{L}_{I_k'}\left(\frac{sT}{\eta _S}\right)-1}{1/\sum_{x  = 0}^{k_s-1}\frac{(-1)^x}{x!}\frac{\partial ^x}{\partial ^{s^x}}\exp\left(\frac{-sT{\sigma^{2} }}{{P\eta}}\right)\mathcal{L}_{I_k}\left(\frac{sT}{\eta _S}\right)-1}.
\end{split}
\end{equation}
\begin{remark}
From the above expressions, it can be concluded that when UE movement is not taken into account, interference is minimal and does not propagate significantly. In this case, the propagation of interference is independent of the density of BSs, RISs, and UEs' behavior. However, when mobility is introduced and based on fixed association rules, increasing the density of BSs, the density of UEs, and the number of RIS elements will lead to stronger interference propagation. Additionally, as will be shown in the performance evaluation section that follows, the intensity of interference propagation increases as more interfering UEs move closer to the target UE. Nonetheless, if the UE's intended signal power is particularly high, the interference's propagation intensity will not be significant.
\end{remark}

\section{Numerical Results and Discussions}
\label{results}
In this section, we present performance evaluation results for the considered multi-RIS-assisted multi-cell downlink wireless system under various conditions, which verify our analytical framework for the system's coverage probability and interference propagation intensity. The role of RISs in interference-limited scenarios, such as cellular communications with cell division, spectrum management, and beamforming, as well as in interference-dominated scenarios, like stadiums and other large venues, spectrum sharing, and UAV Communications, is investigated. Unlike a traditional RIS-assisted communication system, the SINR in an interference-dominated scene may remain below the threshold, and the useful signal power may even be lower than the interference power. Unless otherwise specified, we use the default system parameters as outlined in Table~\ref{tab:t2}.
\begin{table}[!t]
\renewcommand\arraystretch{1} 
\centering
\captionsetup{font={small}}
\caption{Simulation parameters.}
\label{tab:t2} 
\begin{tabular}{|c|c|}    
\hline
\textbf{Parameters} & \textbf{Values} \\
\hline\hline
Bandwidth of the carrier&$10$ MHz\\
\hline
Noise power $\sigma^2$& $-90$ dB\\
\hline
Transmit power P& $[-20,30]$ dBm\\
\hline
PL exponents $\alpha$&$3$\\
\hline
Power allocation coefficients $a_d,a_r$&$0.6,0.4$\\
\hline
Gamma distribution coefficients $m_1,m_2$&$2$\\
\hline
Monte Carlo trials & $10^5$\\
\hline
SINR threshold $T$ & $10^{-2}$\\
\hline
PL of the channel $C$ &$6.3326\times 10^{-5}$\\
\hline
Number of the RIS elements $N$ & 200\\
\hline

Density of BSs ($\lambda_B$), RISs ($\lambda_R$), and UEs ($\lambda_U$) & $10^{-5}$, $10^{-5}$, $10^{-2}$ \\
\hline
\end{tabular}
\end{table}
\subsection{Signal Distribution and Outage Probability}
The cumulative distribution function of the desired signal, including curves obtained by both the derived analytical expression and Monte Carlo simulations, is illustrated in Fig.~\ref{fig:DesiredS}. As observed, the simulation results closely align with the theoretical expression, confirming the accuracy of the signal power distribution model under the studied conditions. The desired signal distribution is primarily influenced by the PL and the remaining channel model parameters, including the Nakagami-$m$ fading model, which characterizes the small-scale fading. The distribution of the desired signal, which represents the total received signal power from both direct and reflected paths, follows a gamma distribution.

The outage probability with and without UE movement as a function of transmit power is shown in Fig.~\ref{fig:OP}. It is evident that this probability decreases as transmit power increases, as expected, since higher transmit power enhances the SINR. Additionally, the UE movement case exhibits a higher outage probability than the no-movement case, especially at lower powers. This increase in outage probability is attributed to interference propagation, in which the UE's movement spreads the interference, thereby reducing the effective SINR.

\begin{figure}[t]
    \centering
        \captionsetup{font={small}}
    \includegraphics[width=1.07\linewidth]{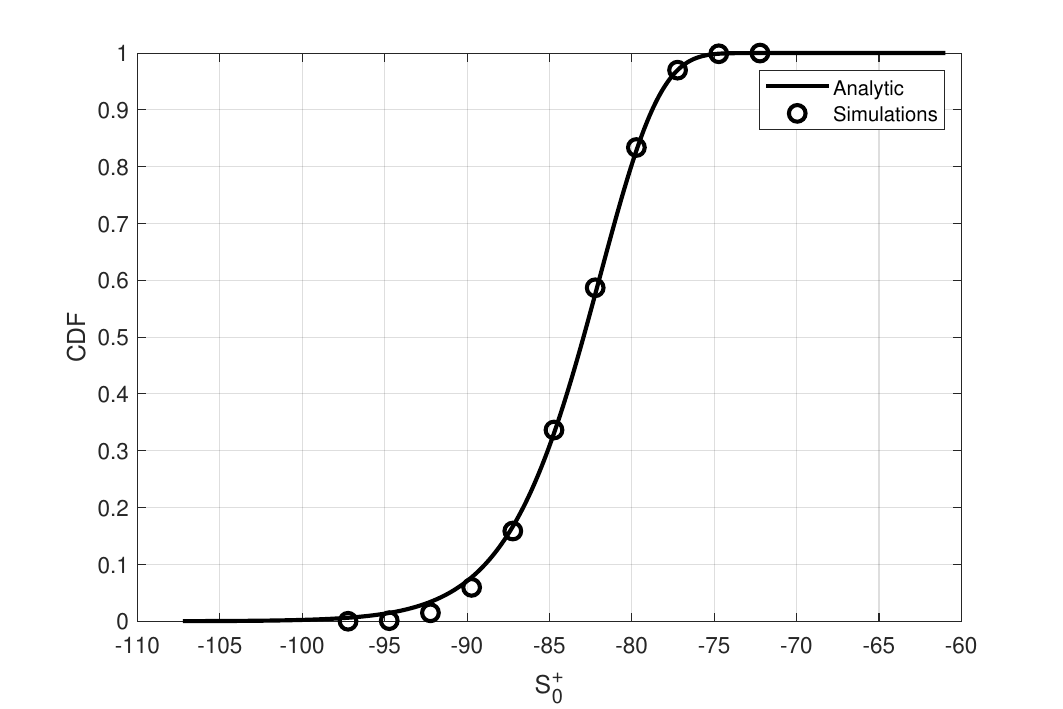}
    \caption{Cumulative distribution function of the desired signal.}
    \label{fig:DesiredS}
\end{figure}

\begin{figure}[t]
    \centering
        \captionsetup{font={small}}
    \includegraphics[width=1.07\linewidth]{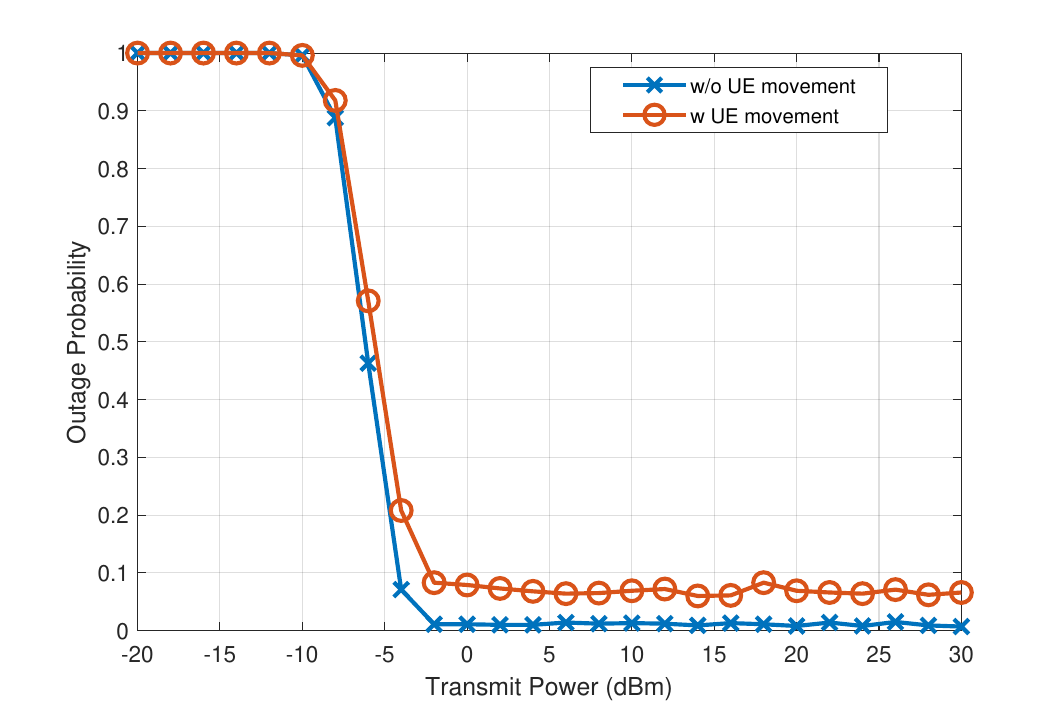}
    \caption{Outage probability with/without UE movement.}
    \label{fig:OP}
\end{figure}

\begin{figure*}[htbp]
    \centering
    \captionsetup{font={small}}
    \begin{subfigure}{0.328\textwidth}
        \centering
        \includegraphics[width=\textwidth]{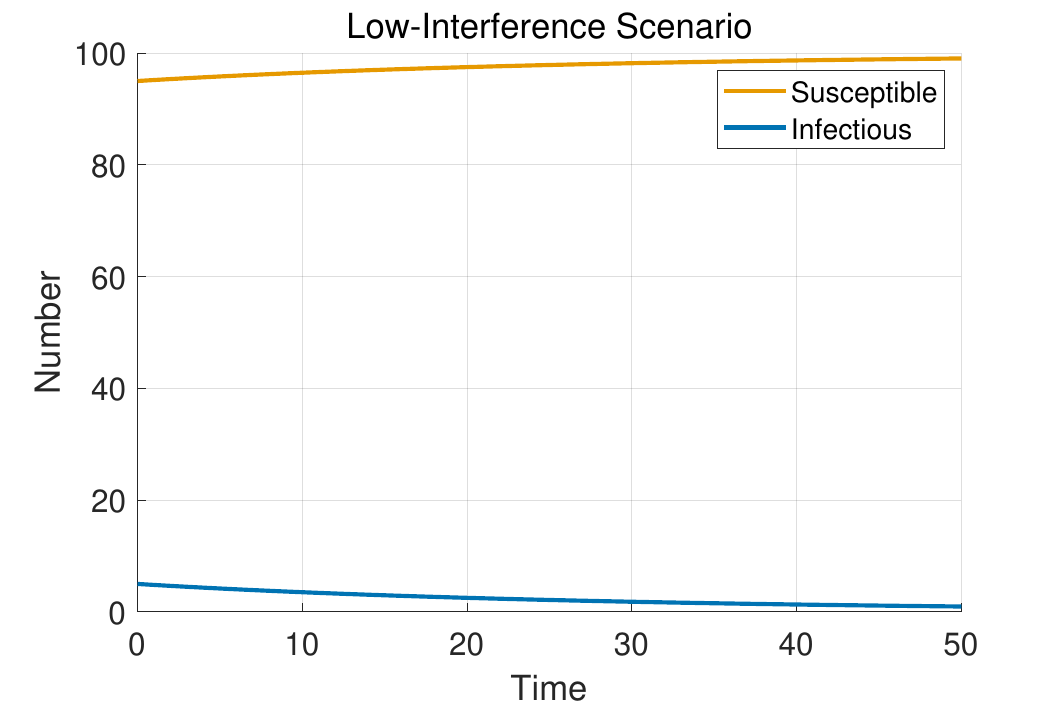}
        \caption{}
        \label{fig:sis1-1}
    \end{subfigure}
    \hfill
    \begin{subfigure}{0.328\textwidth}
        \centering
        \includegraphics[width=\textwidth]{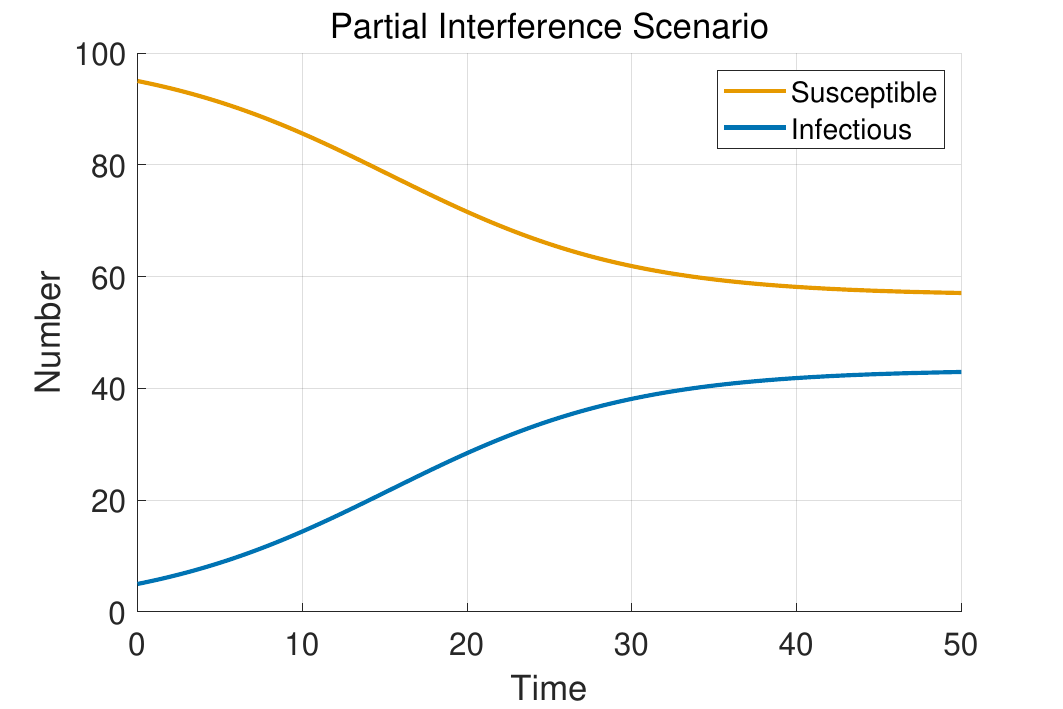}
        \caption{}
        \label{fig:sis1-2}
    \end{subfigure}
    \hfill
    \begin{subfigure}{0.328\textwidth}
        \centering
        \includegraphics[width=\textwidth]{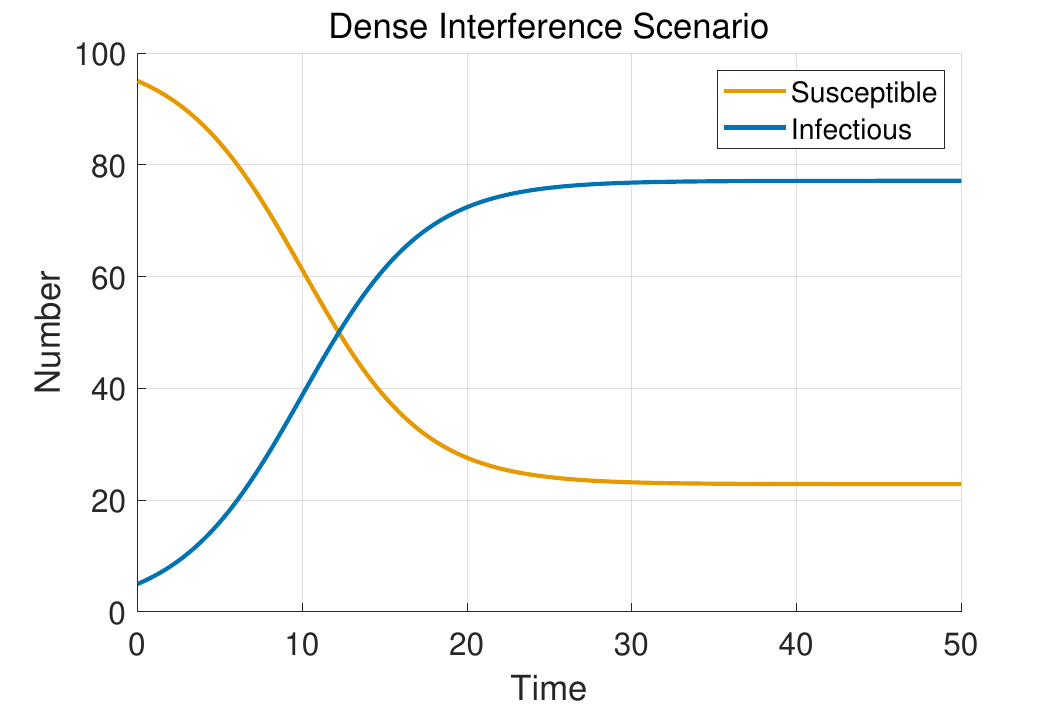}
        \caption{}        
        \label{fig:sis1-3}
    \end{subfigure}
    \hfill
    \begin{subfigure}{0.328\textwidth}
        \centering
        \includegraphics[width=\textwidth]{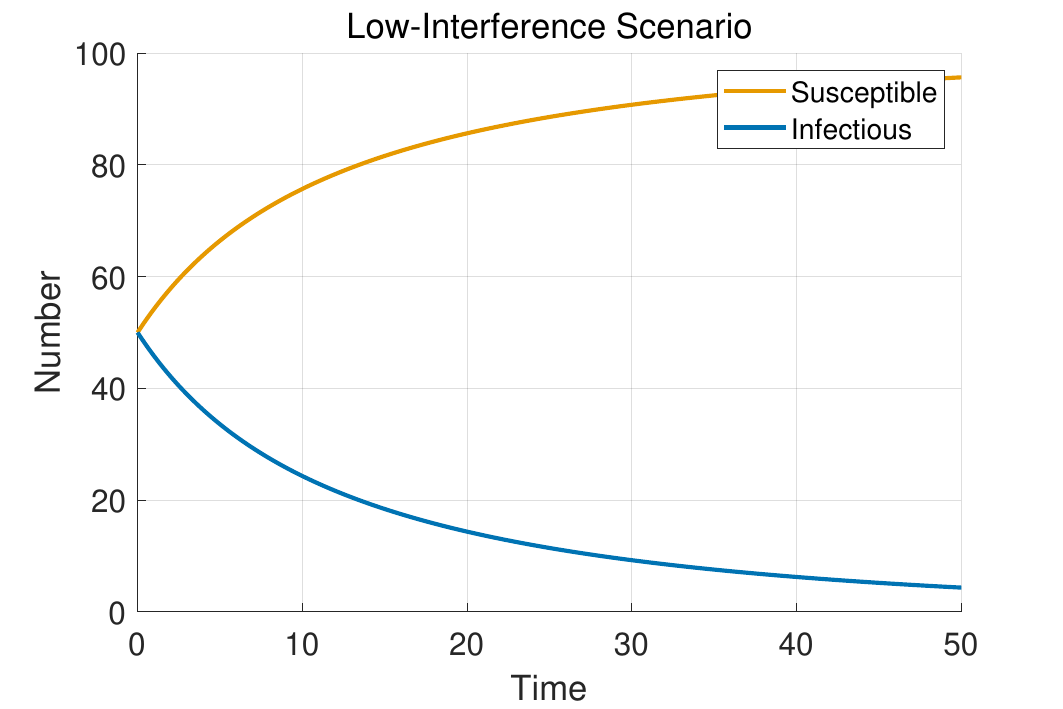}
        \caption{}
        \label{fig:sis2-1}
    \end{subfigure}
    \hfill
    \begin{subfigure}{0.328\textwidth}
        \centering
        \includegraphics[width=\textwidth]{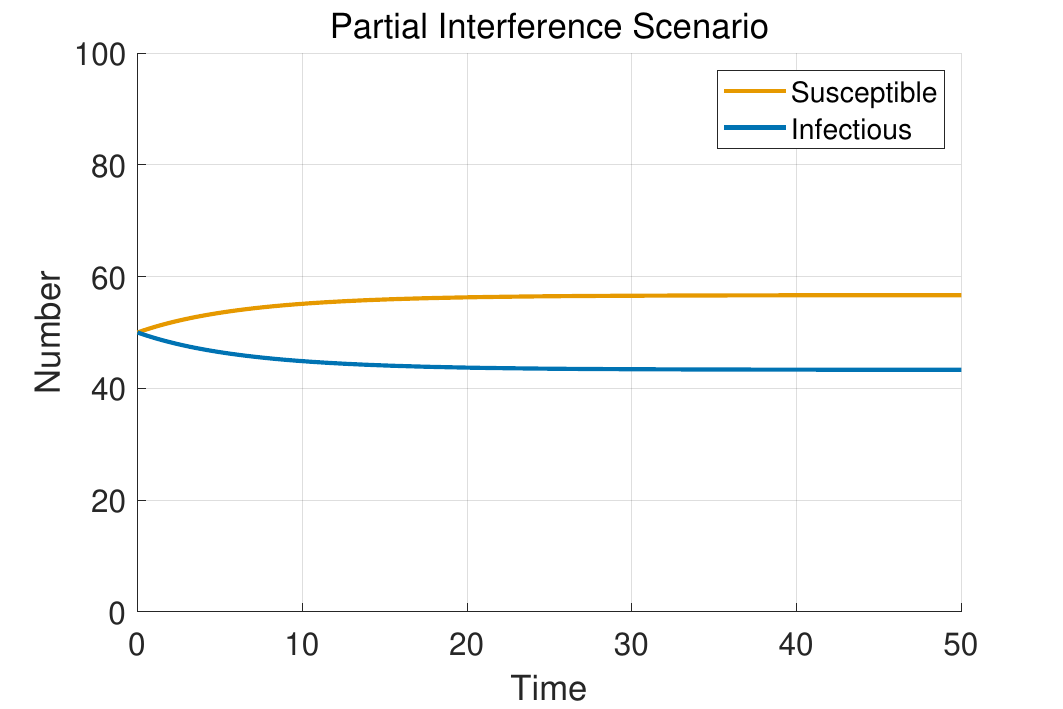}
        \caption{}
        \label{fig:sis2-2}
    \end{subfigure}
    \hfill
    \begin{subfigure}{0.328\textwidth}
        \centering
        \includegraphics[width=\textwidth]{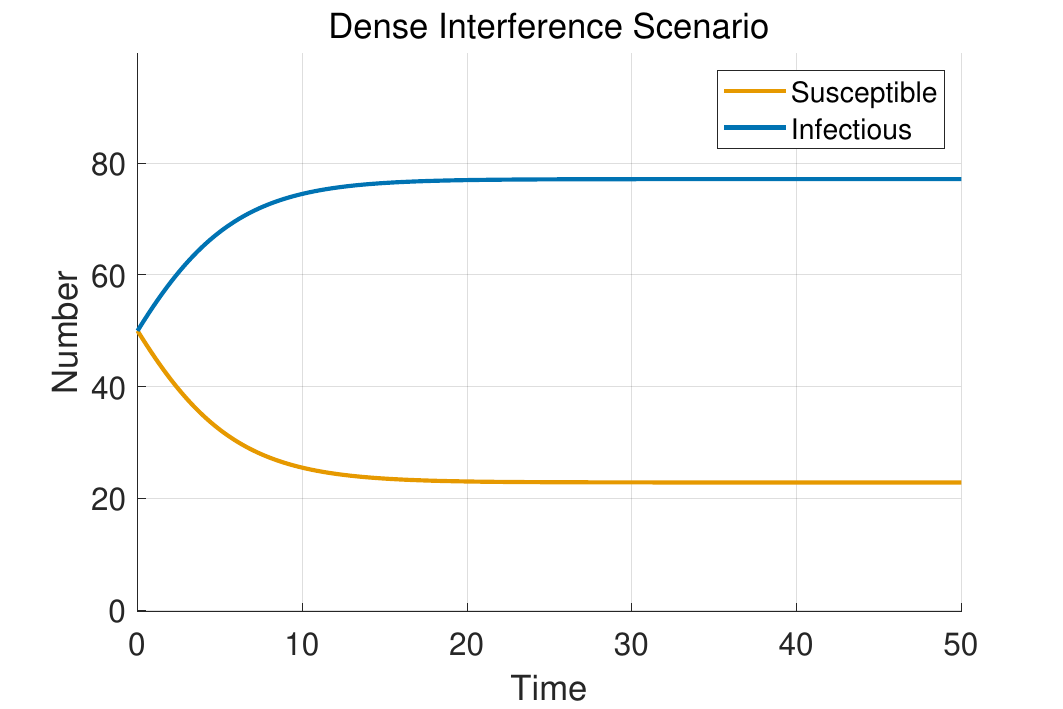}
        \caption{}        
        \label{fig:sis2-3}
    \end{subfigure}
    \caption{Interference propagation under different initial states and interference conditions: low ((a) and (d) subfigures); medium ((b) and (e) subfigures); and high ((c) and (f) subfigures). The initial number of susceptible UEs and the number of infected UEs are $95$ and $5$, respectively, in the top three subfigures, whereas their respective numbers in the bottom three subfigures are $50$ and $50$.
    }
    \label{fig:sis}
\end{figure*}

\subsection{Interference Propagation}
Figure~\ref{fig:sis} shows six subgraphs depicting the interference propagation process under different initial states and interference scenarios, for the case where the transmission power is set to $-5$~dBm. The different initial states refer to the number of infected and susceptible UE at the time the interference in the communication system is evaluated. The interference radius, $r_I$, was set to $10$ meters, meaning that UEs are affected by interference within a $10$-meter radius. The UE densities $\lambda_U$ are $10^{-3}$, $5^{-3}$, and $10^{-2}$, respectively, in Figs.~\ref{fig:sis}(a)-(c). The number of initially infected UEs is $95$, while the number of susceptible UEs is $5$. In Fig.~\ref{fig:sis}(a), where the UE density is low, the number of infected UEs gradually decreases over time, while the number of susceptible UEs gradually increases. This indicates that there are not many interfering UEs in the current scenario, and the target UEs might have moved away from the interfering UEs or the interfering BS. In Fig.~\ref{fig:sis}(b), where the UE density is increased, the number of infected UEs gradually increases, the number of susceptible UEs decreases, and the situation eventually stabilizes. This suggests that during the communication process, some interfering UEs moved away from the target UEs, while others moved closer to them, ultimately leading some UEs to transition from susceptible to infected. However, after this, further movements of these infected UEs do not lead to the infection of new UEs. This also indicates that, in such cases, the interference propagation caused by the RIS is limited. In Fig.~\ref{fig:sis}(c), with a very high UE density, the number of infected UEs continues to increase over time while the number of susceptible UEs decreases. This scenario is typical of crowded places such as amusement parks, stadiums, and shopping malls. The results suggest that when the number of interfering UEs is large and they move rapidly within the scene, interference from RIS reflections propagates repeatedly among different UEs. The findings from Fig.~\ref{fig:sis} are summarized as follows.

\begin{observation}\label{O1}
In sparsely populated communication environments with limited sources of interference, the inherent spatial separation between the UE, combined with the nature of the RIS beam focusing, limits interference propagation. This phenomenon reflects a subcritical state ($R_0 < 1$) \cite{subcritical} of the interference propagation process, analogous to less infectious diseases in epidemic theory. In contrast, in high-density UE deployments, UE mobility and interference beamforming create a supercritical state ($R_0>1$), leading to persistent or even escalating interference propagation. This verifies that RIS-induced interference propagation is a spatially dependent epidemic process.
\end{observation}

\begin{figure}[t]
    \centering    \includegraphics[width=1.09\linewidth]{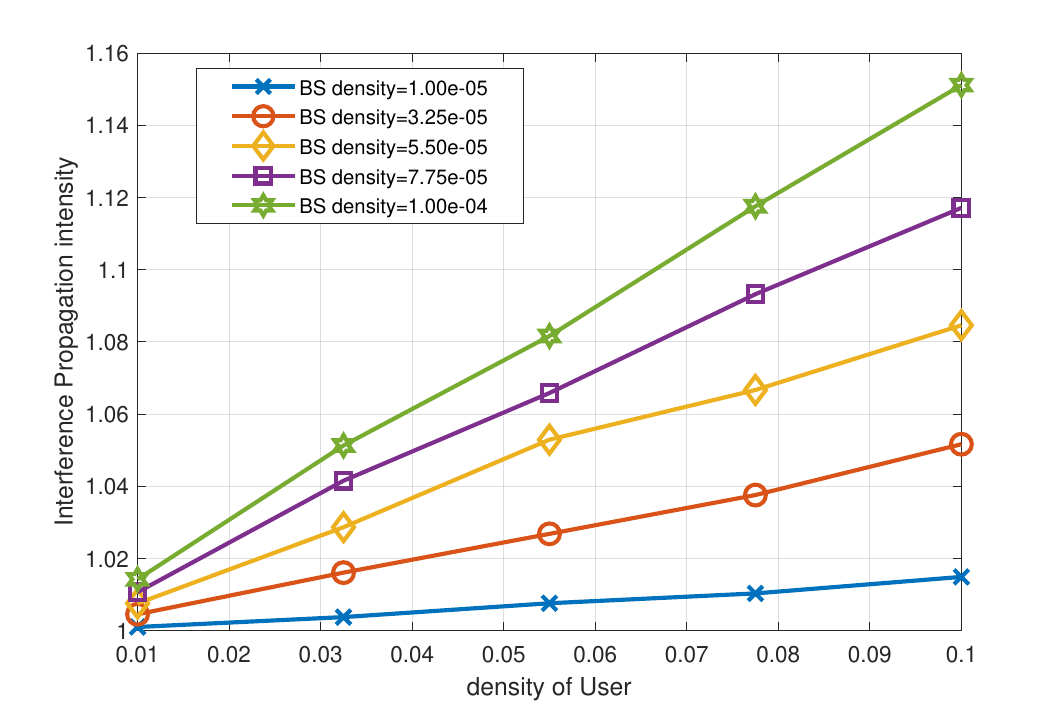}
    \caption{\small Interference propagation intensity versus the density of UEs for different BS density values.}
    \label{regen_densityB}
\end{figure}

Figures~\ref{fig:sis}(d)-(f) provide a comparative perspective to Figs.~\ref{fig:sis}(a)-(c), focusing on scenarios with increased initial infection levels: $50$ infected UEs versus only $5$ in the earlier cases under the same interference settings. This comparison illustrates how the severity of initial interference affects the system’s capacity to amplify interference propagation.
Under a low UE density ($\lambda_U = 10^{-3}$), Fig.~\ref{fig:sis}(d) reveals a key insight: even with a high number of initially infected UEs, the system naturally evolves towards a state with minimal interference. This demonstrates that, even under initially harsh conditions, in sparse networks with low interference coupling, large-scale interference propagation fails to materialize. Figure~\ref{fig:sis}(e) shows that, according to a partial interference scenario with moderate UE density ($\lambda_U = 5^{-3}$), although the higher initial infection level accelerates the early growth in interference, the overall system still stabilizes at a relatively low infection level. This suggests that, in moderately dense environments, interference propagation is possible but self-limiting due to constrained exposure opportunities and limited RIS-induced coupling paths. In contrast, Fig.~\ref{fig:sis}(f) demonstrates the high interference risk in dense UE environments ($\lambda_U = 10^{-2}$). The initially moderate infection level rapidly escalates, and the number of infected UEs quickly saturates near the maximum. This indicates strong propagation effects, facilitated by dense UE and RIS reflections. In such scenarios, e.g., crowded stadiums or public venues, the combination of UE mobility and passive reflection from RISs can result in epidemic-like interference propagation throughout the network. It can thus be concluded that interference control must be an essential feature of densely deployed RIS-assisted communication systems, as summarized in the following sections.

\begin{observation}\label{O2}
The sensitivity of the interference propagation process to initial conditions reveals the nonlinear dynamics in RIS-empowered communication systems. The higher the initial proportion of infected UEs, the faster the convergence to the interference propagation. This is consistent with the phase transition behavior observed in epidemiological SIS models, where higher initial infection densities bypass the substable state and push the system into an equilibrium dominated by persistent disturbances~\cite{PhaseTrans}. Therefore, RIS-assisted communication systems require proactive interference containment strategies, such as adaptive RIS switching strategies \cite{RISswitch} or user-aware scheduling~\cite{10729719}.
\end{observation}

\begin{figure}[t]
    \centering
    \includegraphics[width=1.09\linewidth]{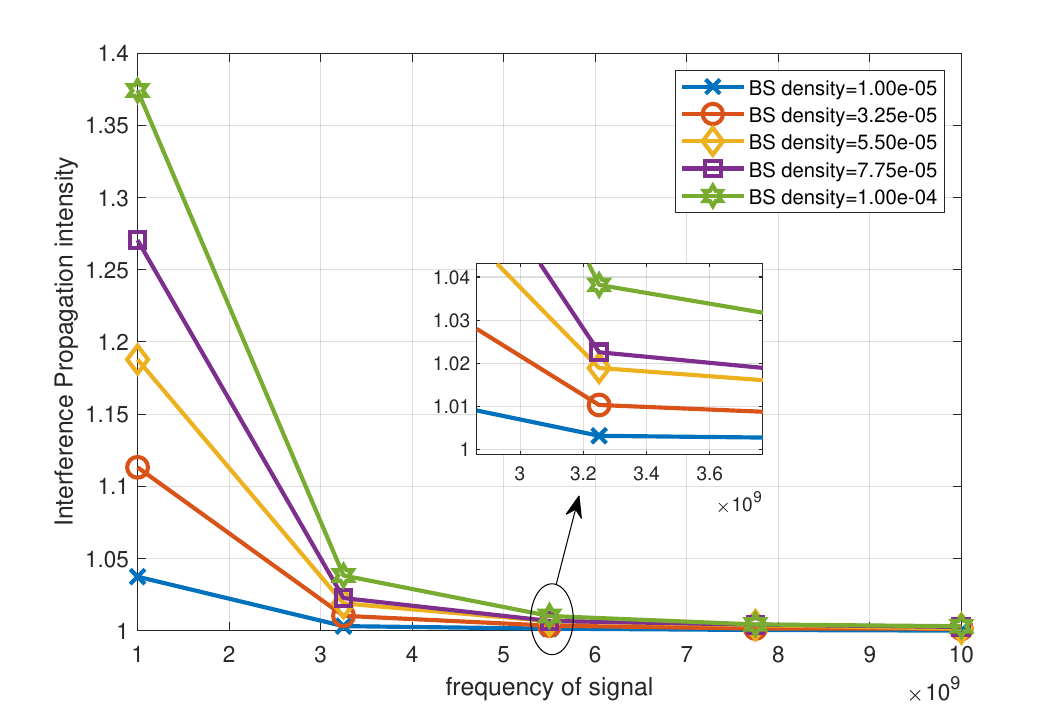}
    \caption{\small Interference propagation intensity versus the signal frequency for different BS density values and low interference conditions.}
    \label{lowI_B_freq}
\end{figure}
\begin{figure}[t]
    \centering
    \includegraphics[width=1.09\linewidth]{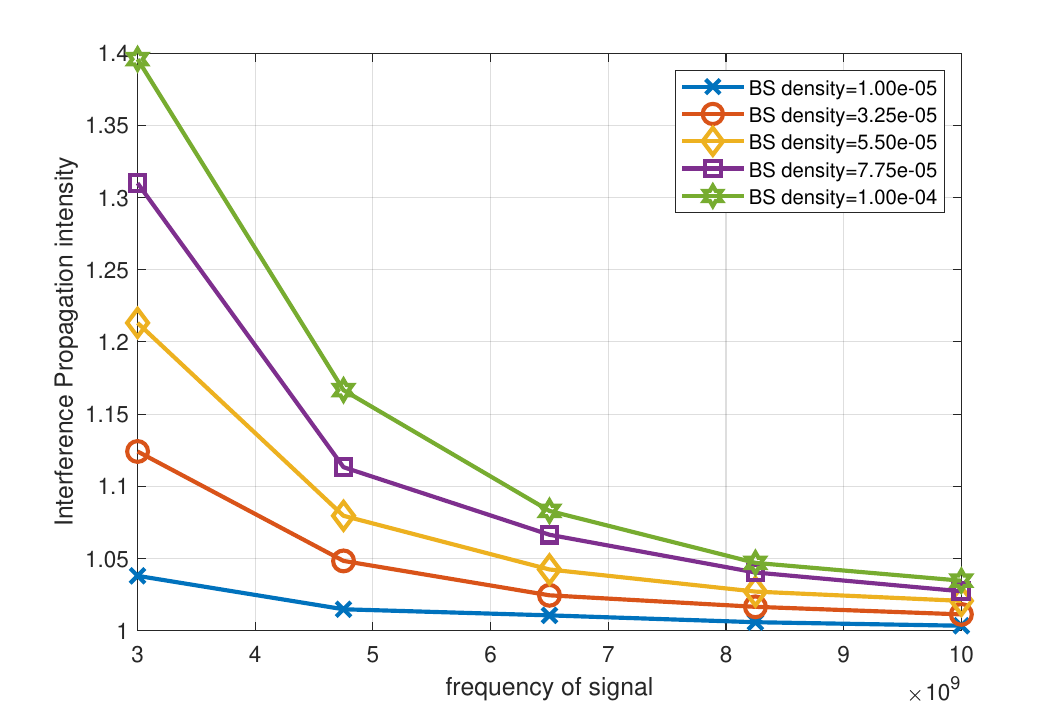}
    \caption{\small Interference propagation intensity versus the signal frequency for different BS density values and high interference conditions.}
    \label{highI_B_freq}
\end{figure}

The impact of different BS densities ($\lambda_B$) and densities of mobile UEs ($\lambda_U$) on the interference propagation intensity is demonstrated in Fig.~\ref{regen_densityB}. As the BS density increases from $10^{-5}$ to $10^{-4}$, the intensity of interference propagation significantly increases. This is mainly because more BSs imply more signal sources, leading to interference signals being reflected and amplified within the RIS-enabled smart wireless environment. In addition, the increase in mobile UE density also improves the efficiency of interference propagation, as more UEs move into the interference area, promoting the propagation of interference signals in the crowd. This phenomenon is particularly evident in densely deployed macro BSs or ultra-dense residential environments in cities, such as smart cities or industrial parks. Although RIS deployment improves signal coverage, high-density BS and UE groups can cause interference to propagate and degrade communication quality. 
\begin{observation} \label{O3}
In system scenarios with high BS and UE densities, it is necessary to reasonably control BS deployment density and optimize RIS reflection configuration to suppress interference propagation. 
\end{observation}

The impact of different signal frequencies and BS densities on the propagation strength of interference in low-interference ($\lambda_U=0.01$) scenarios is shown in Fig.~\ref{lowI_B_freq}. The results indicate that as the signal frequency increases, the intensity of interference propagation decreases. High-frequency signals (such as those in the millimeter-wave band) have shorter wavelengths and stronger beam-focusing in RIS, resulting in more concentrated signals, reduced interference areas, and weaker interference propagation. However, an increase in BS density still leads to increased interference; however, an increase in frequency, to some extent, suppresses the propagation of interference. This conclusion has significant implications for the design of millimeter-wave communication systems.

\begin{observation}\label{O4}
High-frequency communications, e.g., in millimeter-wave and sub-THz bands, offer advantages for mitigating RIS-induced interference propagation. Due to their increased PL and directional narrow-beam transmissions/reflections, higher frequencies reduce the interference exposure region and effectively lower the infection radius in the SIS model. Thus, using higher-frequency bands can not only provide greater bandwidth but also slow the propagation of interference caused by RISs, thereby improving the overall stability of the communication system.
\end{observation}

Figure~\ref{highI_B_freq} includes interference propagation intensity results versus the operating signal frequency in high interference ($\lambda_U=0.1$) scenarios. Evidently, compared with Fig.~\ref{lowI_B_freq}, the interference propagation phenomenon is more significant. This is because more UEs lead to more sources of interference. At the same time, in scenes such as cinemas and sports stadiums, UEs are arranged more neatly, and there are more interfering UEs at the same angle as the target UEs, which makes the interference propagation effect more significant. Similar to Fig.~\ref{lowI_B_freq}, as the signal frequency increases, the interference propagation intensity exhibits a decreasing trend, and an increase in BS density still leads to interference boosting. 
\begin{figure}[t]
    \centering
\includegraphics[width=1.09\linewidth]{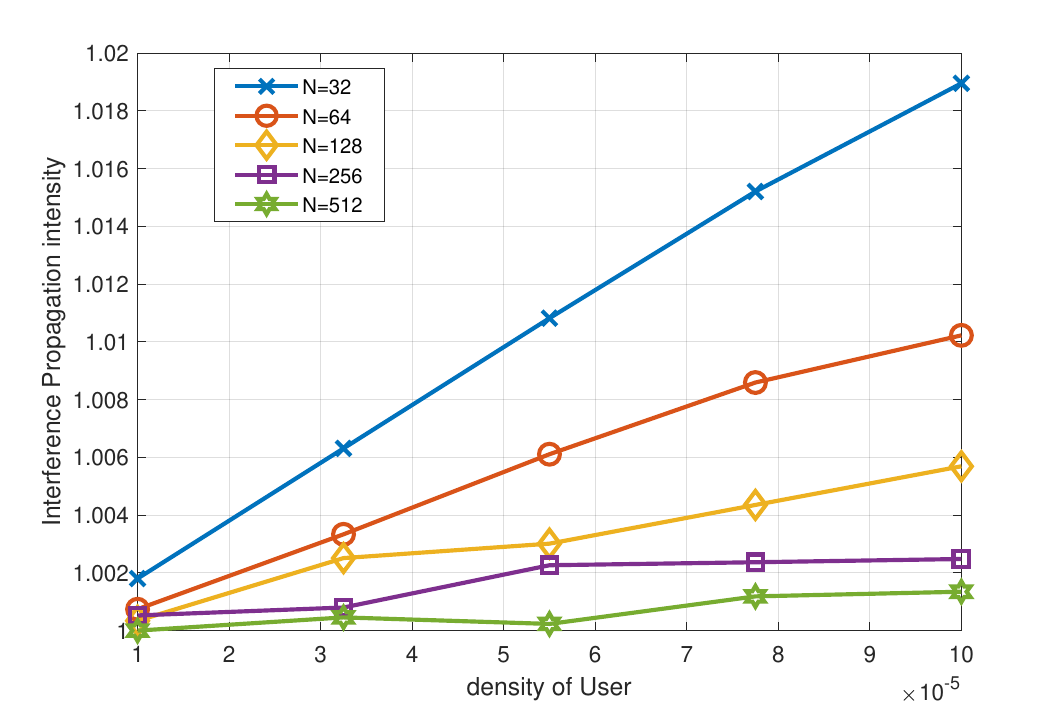}
    \caption{\small Interference propagation intensity versus the density of UEs for different numbers of elements per RIS in low interference scenarios.}
    \label{lowI_B_N}
\end{figure}

The interplay between BS density and the number of RIS elements on interference propagation intensity in scenarios with low user density ($\lambda_U=0.01$) is illustrated in Fig.~\ref{lowI_B_N}. As observed, with increasing RIS element quantity, the interference propagation intensity shows a slight upward trend. However, the overall value remains close to $1$, with a maximum deviation of less than $2\%$. This indicates that, under this configuration, interference propagation can be almost negligible. Specifically, although increasing the density of BS theoretically leads to an increase in sources of interference, the sparse distribution and low mobility of mobile UEs significantly reduce the probability that interference signals will be reflected to the target UEs. Additionally, the limited number of potential infected users and the lack of user movement paths hinder the establishment of an effective interference propagation path. To summarize: 

\begin{observation}\label{O5}
In low mobility and low UE-density scenarios, such as industrial automation areas or sparsely distributed vehicle networks, even with dense infrastructure configurations, RIS deployment will not cause significant interference propagation. This demonstrates that: RISs can be massively deployed to enhance coverage \cite{RIScoverage} without incurring systemic interference, while highlighting the role of UE mobility in managing interference propagation.
\end{observation}

\begin{figure}[t]
    \centering
\includegraphics[width=1.09\linewidth]{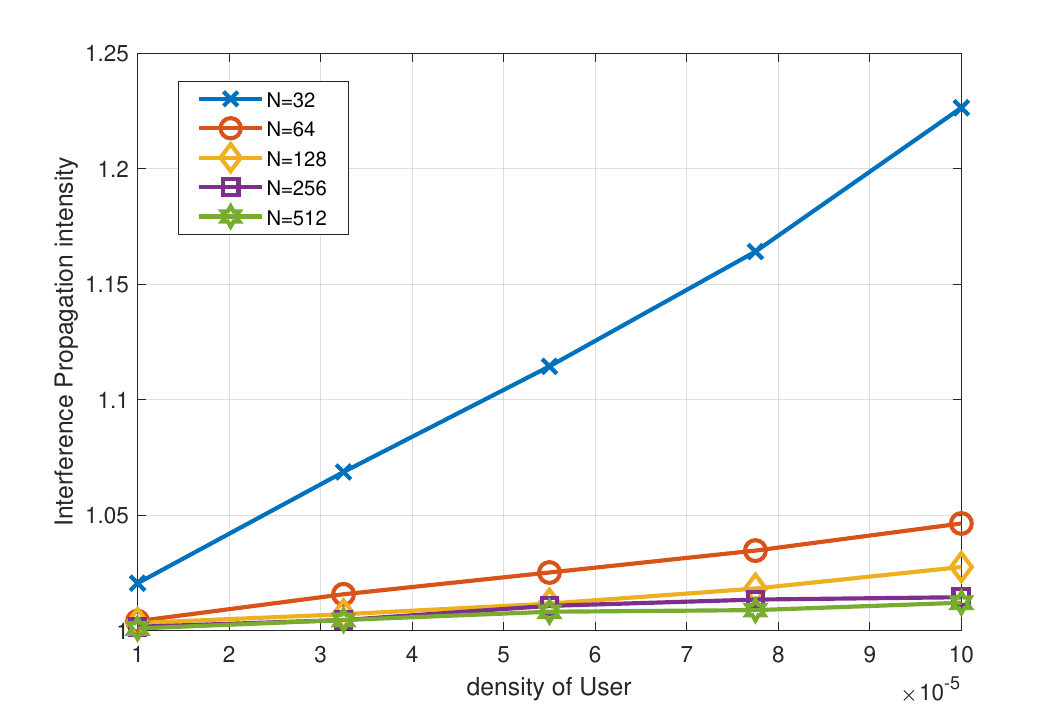}
    \caption{\small Interference propagation intensity versus the density of UEs for different numbers of elements per RIS in high interference scenarios.}
    \label{highI_B_N}
\end{figure}

\begin{figure*}[t]
    \centering
    \begin{align}\label{equa:LP}
        &\mathcal{L}_{I_k}(s)=\mathbb{E}\left\{e^{-sI_k}\right\}=\mathbb{E}_{B,h}\left\{\exp\left(-s\sum_{i\in B\backslash{0}}\left|\sqrt{PL_{ik}}h_{ik}\right|^2\right)\right\}
        \mathbb{E}_{B,R,\phi,h}\left\{\exp\left(-s\sum_{i\in B\backslash \{0\}}\sum_{j\in R}\left|h_{ij}\phi_jh_{jk}\right|^2\right)\right\}\nonumber\\
        &\overset{(a)}{=}\mathbb{E}_{B}\left\{\prod_{i\in {B}\backslash\{0\}} \frac{1}{1+scd_{ik}^{-\alpha }}\right\}
         \mathbb{E}_{B}\prod_{i \in B\backslash \{0\}}\mathbb{E}_{R}\left\{\prod_{j\in R}\frac{1}{1+sNc^2d_{ij}^{-\alpha }d_{jk}^{-\alpha }}\right\}\nonumber\\
        &\overset{(b)}{=}\exp\left(-2\pi\lambda_B\int_0^\infty\left(1-\frac1{1+sc\nu^{-\alpha}}\right)\nu d\nu\right)
        \exp\left(-2\pi\lambda_B\int_0^\infty\left(1-\mathbb{E}_{R}\left\{\prod_{j\in R}\frac1{sNc^2\nu^{-\alpha}d_{jk}^{-\alpha}+1}\right\}\right)\nu d\nu\right)\\
        &\overset{(c)}{=}\exp\left(-2\pi\lambda_B\int_0^\infty\left(1-\frac1{1+sc\nu^{-\alpha}}\right)\nu d\nu\right)
        \exp\left(-2\pi\lambda_B\int_0^\infty\left(1-\left(-2\pi \lambda _R\int_0^\infty\left(1-\frac{1}{sNc^2\nu^{-\alpha}u^{-\alpha }+1}udu\right)\right)\right)\nu d\nu\right)\nonumber\\
        &\overset{(d)}{=}\exp\left(\frac{ -2\pi^2 \lambda_B{\rm csc}\left(\frac{2\pi}{\alpha }\right) (sc)^{2/\alpha}}{\alpha}\right)
\exp\left(-2\pi\lambda_B\left(\frac{2\pi^2\lambda_R{\rm csc}(2\pi/\alpha )sN^2c^2 }{\alpha^2 }(\ln_{}{b}-\ln_{}{a})
        +\frac{\alpha }{\alpha -1}\left(b^{(1-1/\alpha)}-a^{(1-1/\alpha )}\right)\right)\right)\nonumber
    \end{align}
\hrulefill
\end{figure*}

Finally, Fig.~\ref{highI_B_N} illustrates the trend of the influence of BS density and the number of RIS elements on interference propagation intensity in high-interference scenarios ($\lambda_U=0.1$). It can be seen that, as the RISs become larger, the overall interference propagation intensity initially increases and then stabilizes. A larger RIS enhances the signal's reflection and beamforming capabilities, but it also amplifies interference signals through multiple reflections in dense UE deployments, thereby expanding the interference propagation range. Especially in high-density BS-UE scenarios, excessive RISs may exacerbate interference propagation and reduce overall performance. 
\begin{observation}\label{O6}
While reducing the number of RIS elements reduces energy consumption, it also increases the risk of interference in highly dense, high-mobility environments. Therefore, the optimal RIS size must balance the required signal enhancement, energy consumption, and system interference, especially in dynamic system topologies.
\end{observation}

\section{Conclusions}\label{conculsion}
In this paper, we studied interference propagation in multi-RIS-empowered communication systems. We derived novel closed-form expressions for the system's outage probability and the infectious and recovery ratios within the context of the interference propagation model presented. Our findings suggest that while RISs offer significant benefits in terms of coverage and energy efficiency, their deployment can unintentionally lead to interference propagation. We recommend that future RIS deployment strategies explicitly consider mobility patterns, the number of reflecting elements, and user density to fully harness RIS potential. In future work, we will explore more realistic scenarios, including a broader range of SG models. Additionally, we will assess how different categories of RISs and obstacles affect the propagation of interference.


\appendices
\section{}
The Laplace transform of the interference power $I_k$ before UE movement can be written as (\ref{equa:LP}) (top of the next page),
where (a) is due to the independent and identically distributed (i.i.d.) channel power gain $h_{ik}$, $h_{ij}$, $h_{jk}$; (b) and (c) follow from the probability generating functional (PGFL) \cite{PGFL} of PPP according to which the positions of the interferers are modeled by a homogeneous PPP of density; (d) follows the integral solution. The integral over the distance has a divergence solution, so we restrict the upper and lower limits of the integral to $a$, $b$.

\section{}
The Laplace transform of the interference power $I_k$ after UE movement can be expressed as follows:
\begin{align}
    &\mathcal{L}_{I_k'}(s)=\mathbb{E}\left\{e^{-sI_k}\right\}\nonumber\\
    &=\mathbb{E}_{B,h}\left\{\exp\left(-s\sum_{i\in B\backslash{0}}\left|\sqrt{PL_{ik}}h_{ik}\right|^2\right)\right\}\nonumber\\
    &\times\mathbb{E}_{B,R,\phi,h}\left\{\exp\left(-s\sum_{i\in B\backslash \{0\}}\sum_{j\in R}\left|h_{ij}\phi_jh_{jk}\right|^2\right)\right\}\nonumber\\
    &\times \mathbb{E}_{B,R,\phi,h}\left\{\exp\left(-s\sum_{k_I}
    \left|\sqrt{PL_{ijk}}h_{ij}\Phi_j h_{jk}\right|^2\right)\right\}.\label{equa:LP1}
\end{align}
In Appendix A, we have proved the first two terms of $\mathcal{L}_{I_k'}(s)$, and next we will analyze its last term. To this end:
\begin{align}
    & \mathbb{E}_{B,R,\phi,h}\left\{\exp\left(-s\sum_{k_I}
    \left|\sqrt{PL_{ijk}}h_{ij}\Phi_j h_{jk}\right|^2\right)\right\}=\nonumber\\
    &
    =\mathbb{E}_{B}\left\{\prod_{k_I}\frac{1}{1+sNC^2d_{ij}^{-\alpha }d_{jk}^{-\alpha }}\right\}\nonumber\\
    &
    =\exp\left(-2\pi\lambda_{near}\int_0^\infty\left(1-\frac1{1+sNC^2\nu^{-\alpha}d_{jk}^{-\alpha}}\right)\nu d\nu\right)\nonumber\\ 
    &
    =\exp\left(\frac{ -2\pi^2 \lambda_{U_{near}}{\rm csc}\left(\frac{2\pi}{\alpha}\right) (sC)^{2/\alpha}}{\alpha}\right),
\label{equa:LP2}
\end{align}
where similar steps with the derivation of (\ref{equa:LP}) have been used. 

\section{}
According to the definition of $P_o$ in (\ref{PoD}), we have that:
\begin{equation}
\begin{split}
    P_o&=\mathbb{E}\left\{\rm Pr\left[\mathrm{SINR}_0<T\right]\right\}\\
    &=\mathbb{E}\left\{\rm Pr\left[\frac{SP}{IP+\sigma^{2}} <T\right]\right\}\\
    &=\mathbb{E}\left\{\rm Pr\left[S<T\left({I+\frac{\sigma^{2}}{P} }\right)\right]\right\}.
\end{split}
\end{equation}

To solve the above probabilities, we approximate the distribution of $S$ by the Gamma distribution with parameters $k_s$ and $\eta_s$. 
Thus, the probability of $S<T({I+\frac{\sigma^{2}}{P} })$ is
\begin{equation}
\begin{split}
    P_o&=1-\mathbb{E}_{I_k}\left\{\frac{\Gamma\left(k_s,\frac{C(T)}{\eta}\right)}{\Gamma (k_s)}\right\}=1-\sum_{x=0}^{k_s-1}\frac{(-1)^x}{x!}\\
    &\times\mathbb{E}\left\{(-1)^x T\left({{I_k}+\frac{\sigma^{2}}{P} }\right)\exp\left(-T\left({{I_k}+\frac{\sigma^{2}}{P} }\right)\right)\right\}.
\end{split}
\end{equation}
Using the definition of $I_k$'s Laplace transform $\mathcal{L}_{I_k}(s)=\mathbb{E}\{e^{-s{I_k}}\}$, yields the expression:
\begin{equation*}
    P_o=1-\sum_{x  = 0}^{k_s-1}\frac{(-1)^x}{x!}\frac{\partial ^x}{\partial ^{s^x}}\mathbb{E}_{I_k}\left\{\exp\left(\frac{-sT\left({I_k+\frac{\sigma^{2}}{P} }\right)}{\eta}\right)\right\}\Bigg|_{s=1}.
\end{equation*}
This ends the proof.
%
%

\bibliographystyle{IEEEtran}
\bibliography{IEEEabrv,reference}

\begin{thebibliography}{10}
\providecommand{\url}[1]{#1}
\csname url@samestyle\endcsname
\providecommand{\newblock}{\relax}
\providecommand{\bibinfo}[2]{#2}
\providecommand{\BIBentrySTDinterwordspacing}{\spaceskip=0pt\relax}
\providecommand{\BIBentryALTinterwordstretchfactor}{4}
\providecommand{\BIBentryALTinterwordspacing}{\spaceskip=\fontdimen2\font plus
\BIBentryALTinterwordstretchfactor\fontdimen3\font minus \fontdimen4\font\relax}
\providecommand{\BIBforeignlanguage}[2]{{%
\expandafter\ifx\csname l@#1\endcsname\relax
\typeout{** WARNING: IEEEtran.bst: No hyphenation pattern has been}%
\typeout{** loaded for the language `#1'. Using the pattern for}%
\typeout{** the default language instead.}%
\else
\language=\csname l@#1\endcsname
\fi
#2}}
\providecommand{\BIBdecl}{\relax}
\BIBdecl

\bibitem{RIS}
Q.~Wu, B.~Zheng, C.~You, L.~Zhu, K.~Shen, X.~Shao, W.~Mei, B.~Di, H.~Zhang, E.~Basar, L.~Song, M.~Di~Renzo, Z.-Q. Luo, and R.~Zhang, ``Intelligent surfaces empowered wireless network: Recent advances and the road to {6G},'' \emph{Proc. IEEE}, vol. 112, no.~7, pp. 724--763, 2024.

\bibitem{hardware}
E.~Basar, G.~C. Alexandropoulos, Y.~Liu, Q.~Wu, S.~Jin, C.~Yuen, O.~A. Dobre, and R.~Schober, ``Reconfigurable intelligent surfaces for {6G}: Emerging hardware architectures, applications, and open challenges,'' \emph{IEEE Veh. Technol. Mag.}, vol.~19, no.~3, pp. 27--47, 2024.

\bibitem{RIS_challenges}
G.~C. Alexandropoulos, M.~Crozzoli, D.-T. Phan-Huy, K.~D. Katsanos, H.~Wymeersch, P.~Popovski, P.~Ratajczak, Y.~B{\'e}n{\'e}dic, M.-H. Hamon, S.~Herraiz~Gonzalez, R.~D'Errico, and E.~Calvanese~Strinati, ``{RIS}-enabled smart wireless environments: {D}eployment scenarios, network architecture, bandwidth and area of influence,'' \emph{EURASIP J. Wireless Commun. Netw.}, vol. 103, pp. 1--38, 2023.

\bibitem{SINMEC}
X.~Cao, B.~Yang, C.~Huang, C.~Yuen, Y.~Zhang, D.~Niyato, and Z.~Han, ``Converged reconfigurable intelligent surface and mobile edge computing for space information networks,'' \emph{IEEE Network}, vol.~35, no.~4, pp. 42--48, 2021.

\bibitem{RISISACSecure}
C.~Jiang, C.~Zhang, C.~Huang, J.~Ge, D.~Niyato, and C.~Yuen, ``Ris-assisted isac systems for robust secure transmission with imperfect sense estimation,'' \emph{IEEE Transactions on Wireless Communications}, vol.~24, no.~5, pp. 3979--3992, 2025.

\bibitem{HAPWN}
J.~Lyu and R.~Zhang, ``Hybrid active/passive wireless network aided by intelligent reflecting surface: System modeling and performance analysis,'' \emph{IEEE Trans. Wireless Commun.}, vol.~20, no.~11, pp. 7196--7212, 2021.

\bibitem{FGARIS}
L.~Yang, X.~Li, S.~Jin, M.~Matthaiou, and F.-C. Zheng, ``Fine-grained analysis of reconfigurable intelligent surface-assisted mmwave networks,'' in \emph{Proc. IEEE Veh. Technol. Conf}, 2022, pp. 1--6.

\bibitem{RISAMCN}
C.~Zhang, W.~Yi, Y.~Liu, K.~Yang, and Z.~Ding, ``Reconfigurable intelligent surfaces aided multi-cell noma networks: A stochastic geometry model,'' \emph{IEEE Trans. Commun.}, vol.~70, no.~2, pp. 951--966, 2022.

\bibitem{RISMIMOAnalysis}
Y.~Sun, C.-X. Wang, Y.~Xu, L.~Xin, J.~Huang, J.~Huang, Q.~Qin, X.~Gao, B.~Guo, T.~J. Cui, and Y.~Chen, ``Ris-assisted mimo channel measurements and characteristics analysis for 6g wireless communication systems,'' \emph{IEEE Transactions on Vehicular Technology}, vol.~74, no.~9, pp. 13\,335--13\,349, 2025.

\bibitem{SGALIS}
Y.~Zhu, G.~Zheng, and K.-K. Wong, ``Stochastic geometry analysis of large intelligent surface-assisted millimeter wave networks,'' \emph{IEEE J. Sel. Areas Commun.}, vol.~38, no.~8, pp. 1749--1762, 2020.

\bibitem{PARLW}
T.~Wang, G.~Chen, M.-A. Badiu, and J.~P. Coon, ``Performance analysis of {RIS}-assisted large-scale wireless networks using stochastic geometry,'' \emph{IEEE Trans. Wireless Commun.}, vol.~22, no.~11, pp. 7438--7451, 2023.

\bibitem{CAIMW}
A.~Y. Etcibaşı and E.~Aktaş, ``Coverage analysis of {IRS}-aided millimeter-wave networks: A practical approach,'' \emph{IEEE Trans. Wireless Commun.}, vol.~23, no.~4, pp. 3721--3734, 2024.

\bibitem{CARISMC}
L.~Chen, X.~Yuan, and Y.-J.~A. Zhang, ``Coverage analysis of {RIS}-assisted mmwave cellular networks with {3D} beamforming,'' \emph{IEEE Trans. Commun.}, vol.~72, no.~6, pp. 3618--3633, 2024.

\bibitem{RIS-interference}
Y.~Jia, C.~Ye, and Y.~Cui, ``Analysis and optimization of an intelligent reflecting surface-assisted system with interference,'' \emph{IEEE Trans. Wireless Commun.}, vol.~19, no.~12, pp. 8068--8082, 2020.

\bibitem{yb-tvt}
B.~Yang, X.~Cao, C.~Huang, C.~Yuen, L.~Qian, and M.~Di~Renzo, ``Intelligent spectrum learning for wireless networks with reconfigurable intelligent surfaces,'' \emph{IEEE Trans. Veh. Technol.}, vol.~70, no.~4, pp. 3920--3925, 2021.

\bibitem{11212816}
G.~Sun, F.~Baccelli, K.~Feng, L.~U. Garcia, and S.~Paris, ``A stochastic geometry framework for performance analysis of ris-assisted ofdm cellular networks,'' \emph{IEEE Transactions on Wireless Communications}, pp. 1--1, 2025.

\bibitem{10600711}
F.~Saggese, V.~Croisfelt, R.~Kotaba, K.~Stylianopoulos, G.~C. Alexandropoulos, and P.~Popovski, ``On the impact of control signaling in {RIS}-empowered wireless communications,'' \emph{IEEE Open J. Commun. Society}, vol.~5, pp. 4383--4399, 2024.

\bibitem{10670007}
K.~D. Katsanos, P.~Di~Lorenzo, and G.~C. Alexandropoulos, ``Multi-{RIS}-empowered multiple access: A distributed sum-rate maximization approach,'' \emph{IEEE J. Sel. Topics Signal Process.}, vol.~18, no.~7, pp. 1324--1338, 2024.

\bibitem{9693982}
X.~Cao, B.~Yang, C.~Huang, G.~C. Alexandropoulos, C.~Yuen, Z.~Han, H.~V. Poor, and L.~Hanzo, ``Massive access of static and mobile users via reconfigurable intelligent surfaces: Protocol design and performance analysis,'' \emph{IEEE J. Sel. Areas Commun.}, vol.~40, no.~4, pp. 1253--1269, 2022.

\bibitem{SG2}
H.~ElSawy, A.~Sultan-Salem, M.-S. Alouini, and M.~Z. Win, ``Modeling and analysis of cellular networks using stochastic geometry: A tutorial,'' \emph{IEEE Commun. Surveys \& Tuts.}, vol.~19, no.~1, pp. 167--203, 2016.

\bibitem{MHCPP}
H.~He, J.~Xue, T.~Ratnarajah, F.~A. Khan, and C.~B. Papadias, ``Modeling and analysis of cloud radio access networks using matérn hard-core point processes,'' \emph{IEEE Trans. Wireless Commun.}, vol.~15, no.~6, pp. 4074--4087, 2016.

\bibitem{PPP}
R.~W. Heath, M.~Kountouris, and T.~Bai, ``Modeling heterogeneous network interference using poisson point processes,'' \emph{IEEE Trans. Signal Process.}, vol.~61, no.~16, pp. 4114--4126, 2013.

\bibitem{epidemic}
M.~J. Keeling and K.~T. Eames, ``Networks and epidemic models,'' \emph{Journal Royal Society Inter.}, vol.~2, no.~4, pp. 295--307, 2005.

\bibitem{11078147}
G.~C. Alexandropoulos, B.~K. Jung, P.~Gavriilidis, S.~Matos, L.~H. Loeser, V.~Elesina, A.~Clemente, R.~D'Errico, L.~M. Pessoa, and T.~K\"{u}rner, ``Characterization of indoor reconfigurable intelligent surface-assisted channels at 304 {GHz}: Experimental measurements, challenges, and future directions,'' \emph{IEEE Veh. Technol. Mag.}, early access, 2025.

\bibitem{SG3}
M.~Haenggi, \emph{Stochastic Geometry for Wireless Networks}.\hskip 1em plus 0.5em minus 0.4em\relax Cambridge University Press, 2012.

\bibitem{productPL}
J.~Liu and H.~Zhang, ``Performance analysis of {IRS}-assisted networks with product-distance,'' \emph{IEEE Trans. Wireless Commun.}, 2024.

\bibitem{mobility}
J.~Zhang, L.~Fu, X.~Tian, Y.~Cui, and X.~Wang, ``Analysis of random walk mobility models with location heterogeneity,'' \emph{IEEE Trans. Parallel Distribut. Sys.}, vol.~26, no.~10, pp. 2657--2670, 2014.

\bibitem{epidemicModel}
M.~U. Kraemer, J.~L.-H. Tsui, S.~Y. Chang, S.~Lytras, M.~P. Khurana, S.~Vanderslott, S.~Bajaj, N.~Scheidwasser, J.~L. Curran-Sebastian, E.~Semenova \emph{et~al.}, ``Artificial intelligence for modelling infectious disease epidemics,'' \emph{Nature}, vol. 638, no. 8051, pp. 623--635, 2025.

\bibitem{EpidemicInternet}
E.~Tuyishimire and B.~A. Bagula, ``Modelling and analysis of interference diffusion in the internet of things: An epidemic model,'' in \emph{Proc. Conf. Inf. Commun. Technol. Society}, Durban, South Africa, 2020.

\bibitem{EpidemicEmail}
C.~C. Zou, D.~Towsley, and W.~Gong, ``Modeling and simulation study of the propagation and defense of internet e-mail worms,'' \emph{IEEE Trans. Depend. Secure Comput.}, vol.~4, no.~2, pp. 105--118, 2007.

\bibitem{SI}
N.~Mohan and N.~Kumari, ``Positive steady states of a {SI} epidemic model with cross diffusion,'' \emph{Applied Mathematics Comput.}, vol. 410, p. 126423, 2021.

\bibitem{SIS2}
H.~Li and R.~Peng, ``An {SIS} epidemic model with mass action infection mechanism in a patchy environment,'' \emph{Studies in Applied Mathematics}, vol. 150, no.~3, pp. 650--704, 2023.

\bibitem{SIR}
B.~Buonomo and A.~Giacobbe, ``Oscillations in {SIR} behavioural epidemic models: The interplay between behaviour and overexposure to infection,'' \emph{Chaos, Solitons \& Fractals}, vol. 174, p. 113782, 2023.

\bibitem{SEIR}
B.~Marinca, V.~Marinca, and C.~Bogdan, ``Dynamics of {SEIR} epidemic model by optimal auxiliary functions method,'' \emph{Chaos, Solitons \& Fractals}, vol. 147, p. 110949, 2021.

\bibitem{sis}
L.~J. Allen, ``Some discrete-time {SI}, {SIR}, and {SIS} epidemic models,'' \emph{Mathematical biosciences}, vol. 124, no.~1, pp. 83--105, 1994.

\bibitem{mmm}
M.~Telek and G.~Horv{\'a}th, ``A minimal representation of markov arrival processes and a moments matching method,'' \emph{Performance Evaluation}, vol.~64, no. 9-12, pp. 1153--1168, 2007.

\bibitem{nakagami}
M.~D. Yacoub, J.~V. Bautista, and L.~G. de~Rezende~Guedes, ``On higher order statistics of the {N}akagami-$m$ distribution,'' \emph{IEEE Trans. Veh. Technol.}, vol.~48, no.~3, pp. 790--794, 1999.

\bibitem{Interference1}
T.~V. Nguyen, D.~N. Nguyen, M.~Di~Renzo, and R.~Zhang, ``Leveraging secondary reflections and mitigating interference in multi-{IRS/RIS} aided wireless networks,'' \emph{IEEE Trans. Wireless Commun.}, vol.~22, no.~1, pp. 502--517, 2023.

\bibitem{Usermobility}
V.~K. Chapala and S.~M. Zafaruddin, ``Intelligent connectivity through {RIS}-assisted wireless communication: Exact performance analysis with phase errors and mobility,'' \emph{IEEE Trans. Intell. Veh.}, vol.~8, no.~10, pp. 4445--4459, 2023.

\bibitem{phaseError}
Z.~Chu, J.~Zhong, P.~Xiao, D.~Mi, W.~Hao, R.~Tafazolli, and A.~P. Feresidis, ``{RIS} assisted wireless powered {IoT} networks with phase shift error and transceiver hardware impairment,'' \emph{IEEE Trans. Commun.}, vol.~70, no.~7, pp. 4910--4924, 2022.

\bibitem{laplace}
N.~Y. Ermolova and O.~Tirkkonen, ``Laplace transform of product of generalized {M}arcum {Q}, {B}essel {I}, and power functions with applications,'' \emph{IEEE Trans. Signal Process.}, vol.~62, no.~11, pp. 2938--2944, 2014.

\bibitem{subcritical}
G.~Mo and L.~Qiao, ``A molecular dynamics investigation of n-alkanes vaporizing into nitrogen: transition from subcritical to supercritical,'' \emph{Combustion and Flame}, vol. 176, pp. 60--71, 2017.

\bibitem{PhaseTrans}
P.~Van~Mieghem, ``Epidemic phase transition of the {SIS} type in networks,'' \emph{Europhysics Lett.}, vol.~97, no.~4, p. 48004, 2012.

\bibitem{RISswitch}
B.~Yang, X.~Cao, C.~Huang, C.~Yuen, M.~Di~Renzo, Y.~L. Guan, D.~Niyato, L.~Qian, and M.~Debbah, ``Federated spectrum learning for reconfigurable intelligent surfaces-aided wireless edge networks,'' \emph{IEEE Trans. Wireless Commun.}, vol.~21, no.~11, pp. 9610--9626, 2022.

\bibitem{10729719}
A.~L. Moustakas and G.~C. Alexandropoulos, ``{MIMO} mac empowered by reconfigurable intelligent surfaces: Capacity region and large system analysis,'' \emph{IEEE Trans. Wireless Commun.}, vol.~23, no.~12, pp. 19\,245--19\,258, 2024.

\bibitem{RIScoverage}
J.~Sang, Y.~Yuan, W.~Tang, Y.~Li, X.~Li, S.~Jin, Q.~Cheng, and T.~J. Cui, ``Coverage enhancement by deploying {RIS} in {5G} commercial mobile networks: Field trials,'' \emph{IEEE Wireless Commun.}, vol.~31, no.~1, pp. 172--180, 2022.

\bibitem{PGFL}
R.~K. Ganti and M.~Haenggi, ``Asymptotics and approximation of the {SIR} distribution in general cellular networks,'' \emph{IEEE Trans. Wireless Commun.}, vol.~15, no.~3, pp. 2130--2143, 2015.

\end{thebibliography}

\end{document}